\documentclass[12pt]{iopart}

\usepackage[latin1]{inputenc}
\usepackage{iopams}
\usepackage{hyperref}
\usepackage[latin1]{inputenc}
\usepackage{amsmath}
\usepackage{amsthm}
\usepackage{amsfonts}
\usepackage{amssymb}
\usepackage{bm}
\usepackage{enumitem}

\usepackage{float}
\usepackage{graphicx}
\usepackage{makeidx}
\usepackage[caption=false,font=footnotesize]{subfig}
\usepackage[cal=cm,calscaled=1.2,bb=ams,frak=pxtx,frakscaled=1.2,scr=boondox,scrscaled=1.0,bb=boondox]{mathalfa}

\usepackage{breqn}

\usepackage{braket}

\usepackage{framed,color}
\definecolor{shadecolor}{rgb}{0.8,0.8,0.8}

\newtheorem{lemma}{Lemma}[section]
\newtheorem{theorem}{Theorem}[section]
\newtheorem{proposition}{Proposition}[section]

\theoremstyle{remark}
\newtheorem{remark}{Remark}[section]

\theoremstyle{definition}
\newtheorem{definition}{Definition}[section]
\newtheorem{assumption}{Assumption}[section]

\newcommand\norm[1]{\left\lVert#1\right\rVert}
\newcommand\absolut[1]{\left|#1\right|}

\newcommand\trace[1]{\text{Tr}\left[#1\right]}
\newcommand\Paverage[1]{\mathbb{P}\left(#1\right)}

\newcommand\condExpect[1]{\mathbb{E}_{\mathbb{P}}\left[\left.#1\right|\mathscr{Y}_t\right]}
\newcommand\condExpectB[2]{\mathbb{E}_{\mathbb{P}}\left[\left.#1\right|#2\right]}
\newcommand\anticommutator[2]{\left\{#1,#2\right\}}
\newcommand\commutator[2]{\left[#1,#2\right]}

\begin{document}
\title{A Quantum Extended Kalman Filter}
\author{Muhammad F. Emzir, Matthew J. Woolley, and Ian R. Petersen}
\address{School of Engineering and IT , University of New South Wales, ADFA, Canberra, ACT 2600, Australia}
\ead{m.emzir@student.adfa.edu.au}
\vspace{10pt}
\begin{indented}
\item[] March 2016
\end{indented}

\begin{abstract}
	A stochastic filter uses a series of measurements over time to produce estimates of unknown variables based on a dynamic model\cite{anderson1979}. For a quantum system, such an algorithm is provided by a quantum filter \cite{Belavkin1980a}, which is also known as a stochastic master equation (SME)\cite{wiseman2010quantum}. For a linear quantum system subject to linear measurements and Gaussian noise, the quantum filter reduces to a quantum Kalman filter \cite{Doherty1999,yamamoto2006robust}.
	\\
	In this article, we introduce a quantum \emph{extended} Kalman filter (quantum EKF), which applies a commutative approximation and a time-varying linearization to non-commutative quantum stochastic differential equations (QSDEs). We will show that there are conditions under which a filter similar to the classical EKF can be implemented for quantum systems. The boundedness of estimation errors and the filtering problems with `state-dependent' covariances for process and measurement noises are also discussed.
	\\
	We demonstrate the effectiveness of the quantum EKF by applying it to systems which involve multiple modes, nonlinear Hamiltonians and simultaneous jump-diffusive measurements.
\end{abstract}

\section{Introduction}
In light of recent advances in quantum engineering, the need to effectively measure and control complex quantum systems has become more crucial. One requirement is to infer unknown observables of a quantum system from noisy measurements based on a dynamic model, a process known as a filtering. The quantum filter was developed in the 1980's by Belavkin \cite{Belavkin1980a,Belavkin1983,belavkin1989nondemolition,Belavkin1992}. It has recently been used in experimental systems such as trapped ions \cite{Hume2007}, cavity QED systems \cite{sayrin2011real}, and optomechanical systems \cite{Wieczorek2015}.
\\
The real-time application of the quantum filter demands an efficient computational algorithm. The quantum filter in the Schr\"{o}dinger picture, also known as the stochastic master equation, solves the stochastic evolution of the conditional density operator based on the measurement record. In practice, to numerically compute the filter, one has to truncate the Hilbert space basis. The computation time then scales exponentially with the size of the Hilbert space basis which makes the quantum filter difficult to implement in real-time.
\\
In the Heisenberg picture, the quantum filter dynamics reduce to the dynamics of the systems observables $\mathbf{x}_t$. However, for general nonlinear systems, the quantum filter in the Heisenberg picture cannot be interpreted as an explicit solution to the filtering problem \cite{MuhammadF.Emzir2015}.
\\
These two facts lead to the need for new approximation techniques. Among these are a Gaussian approximation of the conditional density operator \cite{Doherty2000,Steck2006}. This work is rather heuristic and lacks an evaluation of the estimation error. Number phase Wigner particle filters have also been suggested\cite{Hush2013}, but like most particle filter based techniques, they are computationally expensive and not suitable for real-time filtering. Recently a Volterra series approximation \cite{tsang2015volterra} was introduced. Although this approach has a tractable error bound \cite{mathews2000polynomial}, estimating the Volterra kernel is complicated in general, and the filter structure is not recursive.
\\
Our approach is to use a commutative operator approximation of the non-commutative nonlinear QSDE. A first order Taylor expansion of the nonlinear quantum Markovian generator is then used to compute the filter gain. We refer to our approach as the quantum extended Kalman filter (quantum EKF), due to the similarity of its structure to the classical extended Kalman filter (EKF). Classical EKF estimation error convergence has been well studied \cite{Ljung1979,Song1992,LaScala1995,Boutayeb1997,Reif2000}, which has resulted in some criteria that guarantee local convergence behavior.
\\
The main difficulty in this approach is that the quantum system, in contrast to a classical system,  is governed by a QSDE which involves non-commutative operators. This non-commutativity implies that there is no isomorphism that could map the dynamics into a standard stochastic differential equation (SDE). Due to the non-commutative nature of the QSDE, generally the ordinary partial differentiation with respect to a system observable may not be well defined \cite{Nazaikinskii1996,Birman2003}. Further, the sub-optimality condition of the estimation error has to be defined differently since the system's observables are operators in Hilbert spaces and not scalar random variables.
\\
In this article, we provide a mathematical description of a quantum EKF. We then establish a sufficient criterion under which the quantum EKF will satisfy a dissipativity condition, which ensures the boundedness of the quadratic estimation error. The cases where the quantum systems and their measurements have a state-dependent covariance, i.e. they are functions of the system observables, are treated as an extension of the quantum EKF. We show that it is sufficient to modify the behavior of the Riccati differential equation to guarantee an equivalent dissipativity condition as the original quantum EKF.
\\
Finally, we demonstrate the application of the quantum EKF to two estimation problems. The first problem is the estimation of quadratures of two cavity modes with Kerr nonlinearities \cite{Walls1994} subject to homodyne detection. The second example is the estimation of quadratures of a cavity subject to simultaneous homodyne detection and photon counting measurements.
\\
The readers are referred to  \cite{bouten2007introduction} for an introduction to quantum probability, quantum stochastic calculus, and quantum non-demolition measurements. We shall use the term `classical' to refer to commutative dynamics and filters. We shall also use the term `state-dependent' covariance to refer to covariance matrices that are functions of the system observables.
\\
The article is organized as follows. The second section will contain some preliminary facts that are used in this article. The main section will describe the main contribution of this work. The first is the mathematical description of the quantum EKF and its existence. The computational complexity of the quantum EKF is analyzed as compared to the SME. Next, we analyze the convergence of the quantum EKF. Lastly, we show an extension of the quantum EKF for state-dependent quantum systems. The third section will be examples of quantum EKF applications, and the last section is the conclusions.

	\subsection{Notation}
	Classical probability spaces are denoted by a triple $(\Omega, \mathscr{F},\mu)$. Plain letters (e.g. $P$) will be used to denote elements of an algebra. $\mathbb{P}$ will be used for a measure from a von Neumann algebra $\mathscr{A}$ to a complex number $\mathbb{C}$, that is positive and normalized, i.e. $\Paverage{A^\ast A} \geq 0$ and $\Paverage{\mathbb{1}}=1$. We also use $\mathbb{E}_{\mathbb{P}}\left(\cdot | \mathscr{A}\right)$ to denote a conditional expectation with measure $\mathbb{P}$ with respect to a commutative von Neumann algebra $\mathscr{A}$. Script face (e.g. $\mathscr{H}$ for Hilbert space) is used to denote a spaces as well as any type of algebra. A class of operators will be denoted by calligraphic face, e.g., for bounded linear operators from a Hilbert space $\mathscr{H}$, to a Hilbert space $\mathscr{K}$, we denote $\mathcal{B}\left(\mathscr{H},\mathscr{K}\right)$, and also we denote $\mathcal{B}\left(\mathscr{H}\right) \equiv \mathcal{B}\left(\mathscr{H},\mathscr{H}\right)$.
	Bold letters (e.g. $\mathbf{y}$) will be used to denote a matrix whose elements are operators on a Hilbert space. Hilbert space adjoints, are indicated by $^{\ast}$, while the complex conjugate transpose will be denoted by $\dagger$, i.e. $\left(\mathbf{X}^{\ast}\right)^{\top} = \mathbf{X}^{\dagger}$. For single-element operators we will use $*$ and $\dagger$ interchangeably. The commutator of $\mathbf{x}$ and $\mathbf{y}$ is given by $[\mathbf{x}, \mathbf{y} ] = \mathbf{x}\mathbf{y}^\top - \left(\mathbf{y}\mathbf{x}^\top\right)^\top$, while their anti-commutator is given by $\anticommutator{\mathbf{x}}{\mathbf{y}}  = \mathbf{x}\mathbf{y}^\top + \left(\mathbf{y}\mathbf{x}^\top\right)^\top$.
	\begin{dmath}
		\alpha
	\end{dmath}
\section{Preliminary}
In the classical stochastic case, the optimal nonlinear filter is given by the Kushner-Stratonovich equation \cite{liptser2001statistics}. This equation is based on the existence of a sample path, $X_t $, whose dynamics is described by a stochastic differential equation, e.g. $dX_t = b(X_t)dt + \sigma(X_t)dW_t$.
The nonlinear filter for a function of $X_t, f(X_t)$, is then given via the generator of the Markov process $f(X_t)$,
\begin{align*}
\mathcal{L}f(X_t) \equiv \lim\limits_{dt\downarrow 0} \dfrac{\mathbb{E}\left[f(X_{t+dt})|X_t = x\right] - f(x)}{dt}.
\end{align*}
In the nonlinear quantum filter, the sample path of the underlying process $X_t$ generally does not exist \cite{bouten2007introduction}, but instead, we have an equivalent evolution of a unitary operator. By means of the evolution of the unitary operator, for open quantum system with Hamiltonian $\mathbb{H}_t$, and bath coupling operator $\mathbb{L}_t$, the equivalent quantum Markovian generator for a set of observables $\mathbf{x}_t$ is given by \cite{lindblad1976generators}
\begin{align}
\mathcal{L}(\mathbf{x}_t) = &  -\dfrac{i}{\hbar}\left[\mathbf{x}_t,\mathbb{H}_t\right] + \frac{1}{2}\left( \mathbb{L}_t^{\dagger} \left[\mathbf{x}_t,\mathbb{L}_t\right]^\top \right)^\top + \frac{1}{2} \left[\mathbb{L}_t^{\ast},\mathbf{x}_t\right]^\top\mathbb{L}_t. \label{eq:QuantumMarkovProcessGenerator}
\end{align}
\\
We begin our formulation by considering Hilbert spaces for the system and the environment. First, let the system Hilbert space and the field Boson Fock space, be given by $\mathscr{H}_s$ and $\Gamma(\mathscr{h})$. The total Hilbert space is given by $\mathscr{H}_{\mathsf{T}} = \mathscr{H}_s \otimes \Gamma(\mathscr{h})$ and $\mathscr{H}_{\mathsf{T}\left[0,t\right]} = \mathscr{H}_s \otimes \Gamma(\mathscr{h})_{\left[0,t\right]}$. The unitary evolution of the system interaction with the field is described by $U_t \in \mathcal{U}(\mathscr{H}_s\otimes\Gamma(\mathscr{h})_{\left[0,t\right]})$, where $\mathcal{U}$ is the class of unitary operators on the associated Hilbert spaces. Let the evolution of a system operator $X$ be given by $X_t = U_t^\ast\left(X\otimes \mathbb{1}\right)U_t$ where $\mathbb{1}$ is identity operator in $\Gamma(\mathscr{h})$. Furthermore, without loss of generality, we will assume that the field is initially on the vacuum state. Let the system and field's initial density operators be given by $\rho \in \mathcal{S}(\mathscr{H}_s) \;, \omega \in \mathcal{S}(\Gamma(\mathscr{h}))$, where $\mathcal{S}$ is the class of unity trace operator on the associated Hilbert spaces.\\
The fundamental quantum processes as in \cite{parthasarathy2012} are given by
$ d\mathbf{A}_t^\ast,d\mathbf{A}_t,d\bm{\Lambda}_t$, the annihilation, creation and conservation processes. These processes are forward time differential, i.e. $d\mathbf{A}_t^\ast,d\mathbf{A}_t,d\bm{\Lambda}_t \in \mathcal{B}(\mathscr{H}_{\mathsf{T}\left[t,t+\delta t\right]})$ and hence commute with $\mathbf{x}_t$. The counting process $d\bm{\lambda}_t$ is defined as the diagonal element of $d\bm{\Lambda}_t$. From now on we let $\mathcal{B} \equiv \mathcal{B}\left(\mathscr{H}_{\mathsf{T}}\right)$ and $\mathcal{B}_{\left[t,t+\delta t\right]} \equiv \mathcal{B}\left(\mathscr{H}_{\mathsf{T}\left[t,t+\delta t\right]}\right)$,  $\mathcal{B}_{\left.t\right]} \equiv \mathcal{B}\left(\mathscr{H}_{\mathsf{T}\left.t\right]}\right)$.
\\
Furthermore, for given  system and field initial density operators $\left(\rho \otimes \omega \right)$ there is a corresponding measure $\mathbb{P}$, called a state, which is positive linear and normalized.
For a bounded operator $X \in \mathcal{B}_{t]}$, $\Paverage{X} = \trace{X \rho \otimes \omega}$, see  \cite[Proposition 9.19]{parthasarathy2012}.
\\
Let us now fix the von Neumann algebra $\mathscr{N} = \mathcal{B}$, and the set of bounded self adjoint operators in total Hilbert space, $\mathcal{O} = \mathcal{O}\left(\mathscr{H}_{\mathsf{T}}\right) \subset \mathscr{N}$. We denote $\mathbf{x}_t \in \mathcal{O}_t$ to be a set of system observables evolved up to time $t$. \\
In the quantum probability setting, a quantum probability space is defined by specifying a von Neumann algebra $\mathscr{N}$ and a state $\mathbb{P}$. Let $\mathscr{A} \subset \mathscr{N}$ be a commutative von Neumann sub algebra. We call the set $\mathscr{A}' = \left\{B \in \mathscr{N}: AB = BA , \;\forall A \in \mathscr{A} \right\}$ the commutant of $\mathscr{A}$ in $ \mathscr{N}$.
The conditional expectation in the quantum probability setting is defined as follows.
\begin{definition}\label{def:QuantumConditionalExpectation}\cite{bouten2007introduction}
	For a given quantum probability space $\left(\mathscr{N},\mathbb{P}\right)$, let $\mathscr{A} \subset \mathscr{N}$ be a commutative von Neumann sub-algebra. Then the map $\condExpectB{\cdot}{\mathscr{A}} : \mathscr{A}' \rightarrow \mathscr{A}$ is called the conditional expectation from $\mathscr{A}'$ onto $\mathscr{A}$ if
	$\Paverage{\condExpectB{B}{\mathscr{A}}A} = \Paverage{BA}, \; \forall A \in \mathscr{A}, \; B \in \mathscr{A}'$.
\end{definition}
The following theorem is fundamental to obtain the relation between quantum probability and the classical Kolmogorov probability axioms. In addition, we will use this theorem to show the implementation of the quantum EKF as a classical EKF in the following section. The proof of the theorem is given in \cite[\textsection 1.1.8]{sakai1971c}\cite{bouten2007introduction}.
\begin{theorem}\label{thm:SpectralTheoremInfiniteDimension}
	Let $\mathscr{A}$ be a commutative von Neumann algebra. Then $\mathscr{A}$ is $\ast$-isomorphic to $\mathcal{L}^\infty\left(\Omega,\mathscr{F},\mu\right)$, with $\ast$-isomorphism $\tau$. Furthermore, a normal state $\mathbb{P}$ on $\mathscr{A}$ defines a probability measure $\mu_{P}$, which is absolutely continuous with respect to $\mu$, such that, $\Paverage{A} = \mathbb{E}_{\mu_{P}}\left(\tau\left(A\right)\right)$, for all $A \in \mathscr{A}$.
\end{theorem}
We notice here that the commutative von Neumann algebra $\mathscr{A}$ corresponds to a classical field $\mathscr{F}$, and every projection operator in  $\mathscr{A}$ corresponds to a classical event.
\\
In the following discussion, we will use $\norm{\cdot}$ the following semi-norm on $\mathcal{B}$: if $\mathbf{x} \in \mathcal{B}^{n\times1}$, then
\begin{dmath}
	\norm{\mathbf{x}} \equiv {\Paverage{\mathbf{x}^\dagger\mathbf{x}}}^{\frac{1}{2}}. \label{eq:SemiNorm}
\end{dmath}
This quantity is not a norm since, $\Paverage{\mathbf{x}^\dagger\mathbf{x}}$ can be zero for non zero $\mathbf{x}$ that is perpendicular to the density operator. Furthermore, a partial order of two operators, $A > B$ is taken in the sense of $\mathbb{P}$, where $A>B$ in $\mathbb{P}$ denotes $\Paverage{A-B} > 0$.
Now suppose we have two operator vectors $\mathbf{x} \in \mathscr{A}', \hat{\mathbf{x}}\in \mathscr{A} $, where $\mathscr{A}$ is a commutative von Neumann algebra. Under semi-norm\eqref{eq:SemiNorm}, by Definition \ref{def:QuantumConditionalExpectation}, we obtain
\begin{dmath}
	\norm{\mathbf{x}-\hat{\mathbf{x}} } =\Paverage{\left(\mathbf{x}- \hat{\mathbf{x}}\right)^\dagger \left(\mathbf{x}- \hat{\mathbf{x}}\right)}^{1/2} = \Paverage{\condExpectB{\left(\mathbf{x}- \hat{\mathbf{x}}\right)^\dagger \left(\mathbf{x}- \hat{\mathbf{x}}\right)}{\mathscr{A}}}^{1/2}. \label{eq:normError}
\end{dmath}
The last equation implies that, for any $\epsilon > 0$, $\left\{\omega : \norm{\mathbf{x}-\hat{\mathbf{x}} } < \epsilon \right\} \in \mathscr{A}$. That is, the event that $\norm{\mathbf{x}-\hat{\mathbf{x}} } < \epsilon$ is $\mathscr{A}-$ measurable. Later on we will use this fact to define a Markov time when we are dealing with the stochastic stability of the quantum EKF, see Section \ref{sec:ConvAnal}.
\\
For a self adjoint element $T \in \mathscr{N}$ there is a $^\ast$-isomorphism $f \rightarrow f(T)$ from a continuous function in the spectrum of $T$, $I_T  = sp(T)\in \mathbb{R}$, namely $\mathcal{C}(I_T)$, onto the $\mathcal{C}^\ast$-subalgebra $\mathcal{C}^\ast (T)$ of $\mathscr{N}$ generated by $T$ and the identity element $\mathbb{1}$.
\\
The following two propositions show how one can define a partial derivative for an operator differentiable mappings which will be used in the filter algorithm. Both have been proved in \cite{pedersen2000operator,Widom1983}, but we mention the proof here again for the sake of completeness. For the two propositions, we will recall the following definitions,
\begin{definition}\cite{hutson2005applications,pedersen2000operator}If $\mathscr{X}$ and $\mathscr{Y}$ are Banach spaces with norm $\norm{\cdot}_{\mathscr{X}}$ and $\norm{\cdot}_{\mathscr{Y}}$ respectively, a mapping $f : \mathscr{D} \rightarrow \mathscr{Y}$ on a subset $\mathscr{D}$ of $\mathscr{X}$ is Fr\'{e}chet differentiable at $T \in \mathscr{D}$ if there is a bounded linear operator $D_{\left(f,T\right)}$ in $\mathcal{B}\left(\mathscr{X},\mathscr{Y}\right)$, the class of linear bounded functions from $\mathscr{X}$ to $\mathscr{Y}$,  such that,
	\begin{dmath}
		\lim\limits_{S\rightarrow 0} \dfrac{\norm{f(T+S) - f(T) - D_{\left(f,T\right)}S}_{\mathscr{Y}}}{\norm{S}_{\mathscr{X}}} = 0. \label{eq:FrechectDifferential}
	\end{dmath}
	If $D_{\left(f,T\right)}$ is defined for every $T \in \mathscr{X}$, then $f$ is Fr\'{e}chet differentiable on $\mathscr{D}$.\label{def:FrechetDifferentiable}
\end{definition}
\begin{definition}\cite{pedersen2000operator}
	Let $\mathscr{U}$ be any unitary C-* algebra, and $\mathscr{S}$ is the corresponding self adjoint sub-algebra. For any continuous function $f$ on the compact interval $I$, $f \in \mathcal{C}\left(I\right)$ is said to be operator differentiable if the operator function
	\begin{align*}
	f : \mathscr{S}^{I} \rightarrow \mathscr{U},
	\end{align*}
	is Fr\'{e}chet differentiable on $\mathscr{D} = \mathscr{S}^{I}$, symbolically $f \in \mathcal{C}^1_{op}(I)$.
\end{definition}

\begin{proposition}\cite{Widom1983,pedersen2000operator}
	If $f \in  \mathcal{C}^1_{op}(I)$, then $f \in \mathcal{C}^1\left(I\right)$
	\label{prp:FinCpOPthenFBelongsToC1}
\end{proposition}
\begin{proof}
	Without losing of generality, we may assume that $I$ is bounded, i.e $I = \left\{x\in \mathbb{R}: a \leq i \leq b \right\}$, and $f:\mathscr{U} \rightarrow \mathscr{U}$. Set $\mathscr{U} = \mathcal{C}_b\left(I\right)$, the set of a bounded continuous functions on $I$. Since $f \in \mathcal{C}^1_{op}(I) $, the differentiability of $f$ at a function $g \in \mathscr{U}$ implies that for every $\epsilon > 0$, there exists $\delta > 0$ such that for any function $h, \norm{h} < \delta$
	\begin{dmath}
		\norm{f(g+h)-f(g) - D_{\left(f,g\right)} h } \leq \epsilon \norm{h},
	\end{dmath}
	which shows that, $D_{\left(f,g\right)} h : h \rightarrow \left(f' \circ g\right) h $. Suppose there exist two arbitrary points $x_1,x_0 \in I$, satisfying $\absolut{x_1-x_0}<\delta$. Now let $g(x)$ equal the identity function, and $h(x) = x_0-x_1$ constant. We get, $\norm{h} < \delta$, and at $x=x_1$,
	\begin{dmath*}
		f(x_1+x_0-x_1)-f(x_1) - f'\left(x_1\right) (x_0-x_1) = f(x_0)-f(x_1) - f'\left(x_1\right) (x_0-x_1).
	\end{dmath*}
	Next interchanging $x_0$ and $x_1$, we get,
	\begin{dmath*}
		f(x_1)-f(x_0) - f'\left(x_0\right) (x_1-x_0).
	\end{dmath*}
	Adding this and the previous equation, and dividing by $\norm{x_1-x_0}$, we get
	\begin{dmath*}
		\norm{f'\left(x_0\right)- f'\left(x_1\right)} < \epsilon
	\end{dmath*}
	which shows the continuity of $f'$ on $I$.
\end{proof}

\begin{proposition}\cite{pedersen2000operator}
	If $f \in \mathcal{C}^1_{op}(I)$, for any two elements $S,T \in \mathscr{U}$, a unital commutative C-$^\ast$ algebra,  then
	\begin{align*}
	D_{\left(f,T\right)}S = f'(T)S.
	\end{align*}
	\label{prp:OrdinaryDerivative}
\end{proposition}
\begin{proof}
	Let $S,T \in \mathscr{U}$, a unital commutative C-$^\ast$ algebra on the space $\mathscr{X}$. We write, $f(T)(x) = f(T(x)), \forall x \in \mathscr{X}$, so that, as the previous proposition, $f(T) = f \circ T$. Since $f \in \mathcal{C}^1_{op}(I)$, by definition, there exists a differential $D_{T}f \in  \mathcal{B}\left(\mathscr{U},\mathscr{U}\right)$, such that,
	\begin{dmath}
		f(T + \epsilon S) (x) = f(T ) (x) + \epsilon D_{\left(f,T\right)}S(x) +  \epsilon R_{\epsilon} (x),
	\end{dmath}
	where $\lim\limits_{\epsilon \rightarrow 0} R_{\epsilon} = 0$. By Proposition \ref{prp:FinCpOPthenFBelongsToC1}, it follows that,
	\begin{dmath*}
		D_{\left(f,T\right)}S(x) = f'(T(x))S(x),
	\end{dmath*}
	as desired.
\end{proof}
By definition the last proposition is also true for a commutative von Neumann algebra, since any von Neumann algebra is a C-$^\ast$ subalgera of $\mathcal{B}\left(\mathscr{H}\right)$ that is strongly closed and contains the identity, for $\mathscr{H}$ is any Hilbert space \cite{conway2000course}.
\\
The condition that $\mathscr{U}$ is commutative in Proposition \ref{prp:OrdinaryDerivative} is essential. In addition, it was proven that one can construct a $\mathcal{C}^1$ function that is not operator differentiable \cite{farforovskaya1975example}. Moreover, if $f \in \mathcal{C}^2$, then its extension is operator differentiable, see \cite{arazy1990operator,kissin1996operator}. Hence $ \mathcal{C}^2 \subseteq \mathcal{C}_{op}^1 \subseteq \mathcal{C}^1$.

\begin{definition}
	An operator $A$ in a Hilbert space $\mathscr{H}$ is single valued, if it can be written as a scalar times an identity,
	\begin{align}
	A = a \mathbb{1}, a \in \mathbb{C}. \label{eq:HatSingleValued}
	\end{align}
	A matrix $\mathbf{A}$ whose each elements are operators in $\mathscr{H}$ is single valued if each element of $\mathbf{A}$ is single valued.
\end{definition}

	\section{Main Statement}
	\subsection{Extended Kalman filter for a class of open quantum quantum systems with diffusive and Poissonian measurements}

	Let $\mathcal{G}$ be an $m$ channel open quantum system with parameters $(\mathbb{S},\mathbb{L},\mathbb{H})$, $\mathbb{S}\in \mathbb{C}^{m\times m}, \mathbb{L} \in \mathcal{B}_t^{m\times 1}, \mathbb{H} \in \mathcal{O}_t^{m\times 1}$  and $\mathbb{S}\mathbb{S}^{\dagger} = \mathbb{S}^{\dagger}\mathbb{S} = \mathbf{I}$.
	The class of open quantum system that we consider here is the one which has covariance matrices independent of $\mathbf{x}_t$. A more general class of open quantum system with state-dependent covariances will be considered in the following section.
	The smallest von Neumann algebra generated from the measurement up to time $t$ is given by $\mathscr{Y}_t$.
	As in the classical setting, we can always assume that the von Neumann algebra $\mathscr{Y}_t$ is right continuous, by taking $\mathscr{Y}_t = \mathscr{Y}_{t_+} = \bigcap_{s>t}\mathscr{Y}_t$ \cite{klebaner2005introduction}.
	\\
	We assume that the scattering matrix $\mathbb{S}$ is single valued. This implies that for any $X_t = U_t^\dagger X U_t$, $\trace{( (\mathbb{S}^{\dagger}X_t\mathbb{S})-X_t)d\bm{\Lambda}_t^\top}  = 0$. This ensures that the open quantum system evolution can be described by a diffusive QSDE as follows\cite{gough2009series}:
	\begin{dgroup}
		\begin{dmath}
			d\mathbf{x}_t = \mathbf{f}(\mathbf{x}_t)dt + \mathbf{G}(\mathbf{x}_t) d\mathbf{A}_t^\ast + \mathbf{G}(\mathbf{x}_t)^\ast d\mathbf{A}_t, \label{eq:QSDE}
		\end{dmath}
		\begin{dsuspend}
			with,
		\end{dsuspend}
		\begin{dmath*}
			\mathbf{f}(\mathbf{x}_t) =  \mathcal{L}(\mathbf{x}_t),\label{eq:fQSDE}
		\end{dmath*}
		\begin{dmath*}
			\mathbf{G}(\mathbf{x}_t) =  \left[\mathbf{x}_t,\mathbb{L}_t\right] \mathbb{S}_t^\ast.\label{eq:G}
		\end{dmath*}

	\end{dgroup}
	We shall assume that the number of measurements subject to the open quantum system $\mathcal{G}$ also equal to $m$.
	The class of measurements of the quantum system above are assumed to be a collection of $m$ functions of output field creation, annihilation, and conservation processes, as below
	\begin{dgroup}
		\begin{dmath}
			d\mathbf{y}_t = \mathbf{E}_td\tilde{\mathbf{A}}_t^{\ast} +\mathbf{E}^{\ast}_t d\tilde{\mathbf{A}}_t+ \mathbf{N}_t d\tilde{\bm{\lambda}}_t,\label{eq:Measurements}
		\end{dmath}
		\begin{dsuspend}
			where,
		\end{dsuspend}
		\begin{dmath*}
			d\tilde{\bm{\lambda}}_t = \text{diag}\left(d\tilde{\bm{\Lambda}}_t\right),
		\end{dmath*}
		\begin{dmath*}
			d\tilde{\mathbf{A}}_t = \mathbb{S}_td\mathbf{A}_t + \mathbb{L}_t dt,
		\end{dmath*}
		\begin{dmath*}
			d\tilde{\bm{\Lambda}}_t = \mathbb{S}_t^{*} d\bm{\Lambda} \mathbb{S}_t^{\top} + \mathbb{S}_t^{*} d\mathbf{A}_t^{\ast}\mathbb{L}_t^{\top} + \mathbb{L}_t^{\ast}d\mathbf{A}_t^{\top}\mathbb{S}_t^{\top} + \mathbb{L}^{\ast}\mathbb{L}^{\top}dt\label{eq:CountingMeasurement} .
		\end{dmath*}
	\end{dgroup}
	To satisfy the \emph{non-demolition} and \emph{self-nondemolition} properties, $\mathbf{E}_t$ and $\mathbf{N}_t$ have to satisfy the algebraic condition in  Theorem 3.1 of \cite{MuhammadF.Emzir2015}. We could then simplify \eqref{eq:Measurements} to
	\begin{subequations}
		\begin{align}
		d\mathbf{y}_t =& \mathbf{h}(\mathbf{x}_t)dt + \mathbf{L}(\mathbf{x}_t) d\mathbf{A}_t^\ast + \mathbf{L}(\mathbf{x}_t)^\ast d\mathbf{A}_t + \mathbf{N}_t d\bm{\alpha}_t, \label{eq:nonLinearMeasurementQSDE} \\
		\mathbf{h}(\mathbf{x}_t) =& \mathbf{E}^{\ast}_t\mathbb{L}_t+\mathbf{E}_t\mathbb{L}_t^{\ast}+\mathbf{N}_t \mathbf{l}_t, \\
		\mathbf{L}(\mathbf{x}_t) =& \left(\mathbf{E}_t+\mathbf{N}_t
		\bar{\mathbb{L}}\right)\mathbb{S}_t^{*},
		\end{align}\label{eq:MeasurementsSimplified}
	\end{subequations}
	where $\mathbf{N}_t \in \mathcal{O}^{m\times m}$, and
	\begin{dgroup*}
		\begin{dmath*}
			{\bar{\mathbb{L}}  = \begin{bmatrix}
					\mathbb{L}_{1,t} && \cdots  && 0\\
					0 && \mathbb{L}_{i,t}  && 0\\
					0 && \cdots  && \mathbb{L}_{m,t}\\
				\end{bmatrix}},
				{\mathbf{l}_t  =  \begin{bmatrix}
						\mathbb{L}_{1,t}^\ast\mathbb{L}_{1,t}\\
						\vdots\\
						\mathbb{L}_{m,t}^\ast\mathbb{L}_{m,t}\\
					\end{bmatrix}
				},\\
				{d\bm{\alpha}_t =  \text{diag}\left(\mathbb{S}_t^{*} d\bm{\Lambda} \mathbb{S}_t^{\top}\right).
				}
			\end{dmath*}
		\end{dgroup*}
		Now, we define the variance and covariance of the system's observables and measurements as follows,
		\begin{dgroup}
			\begin{dmath}
				\mathbf{Q}_t = \frac{1}{2 dt}\condExpect{\anticommutator{d\mathbf{x}_t}{d\mathbf{x}_t} }, \label{eq:Q}
			\end{dmath}
			\begin{dmath}
				\mathbf{R}_t = \frac{1}{2 dt} \condExpect{\anticommutator{d\mathbf{y}_t}{d\mathbf{y}_t} }, \label{eq:R}
			\end{dmath}
			\begin{dmath}
				\mathbf{S}_t = \frac{1}{2 dt} \condExpect{\anticommutator{d\mathbf{x}_t}{d\mathbf{y}_t} }. \label{eq:S}
			\end{dmath}
			\label{eq:Variances}
		\end{dgroup}
		The use of anti-commutator in \eqref{eq:Variances} above ensures that each element of the variance matrices is a self adjoint operator. In the remainder of the article, we require that $\mathbf{R}_t$ is a positive definite matrix of operators in $\mathscr{Y}_t$, i.e $\tau\left(\mathbf{R}\right)_{t,\omega} > 0, \; \forall t\geq 0, \omega \in \Omega$, where $\tau$ is the $\ast-$-isomporhism from $\mathscr{Y}_t$ to $\mathcal{L}^\infty\left(\Omega,\mathscr{F},\mu\right)$.
		\\
		We will use Proposition \ref{prp:OrdinaryDerivative} to calculate the partial derivative of the nonlinear quantum Markovian process generator $f(\mathbf{x})$. As we will show shortly, we could establish a mild condition on the quantum systems parameters $\mathbb{S},\mathbb{L},\mathbb{H}$ so that $f_i,h_j\in \mathcal{C}^1_{op}(\mathbb{R}) \forall i \leq n, j \leq m$. The following proposition gives a sufficient condition under which $f_i,h_j\in \mathcal{C}^1_{op}(\mathbb{R}) \forall i \leq n, j \leq m$ belongs to $ \mathcal{C}^1_{op}(\mathbb{R})$.
		\begin{proposition}
			In order to have $f_i,h_j\in \mathcal{C}^1_{op}(\mathbb{R}) \forall i \leq n, j \leq m$, it is sufficient to require that $\mathbb{H}, \mathbb{L} \in \mathcal{C}^1_{op}(\mathbb{R})$ and $\mathbf{E}_t,\mathbf{N}_t\in \mathcal{C}^1_{op}(\mathbb{R})$
		\end{proposition}
		\begin{proof}
			To prove this, we will first claim that for any two function, $g_1,g_2 \in \mathcal{C}^1_{op}(I)$, with $I\subseteq \mathbb{R}$ , then $g_1g_2 \in  \mathcal{C}^1_{op}(I)$. Without loss of the generality, we can fix any $S \in \mathscr{U}$, where $\norm{S}_{\mathscr{U}}=1$. Letting $\varepsilon \rightarrow 0$, the linear operator $D_{g_i,T}$ as in \eqref{eq:FrechectDifferential} is given by
			\begin{align*}
			D_{g_i,T} =& \lim\limits_{\varepsilon \rightarrow 0} \dfrac{g_i\left(T+\varepsilon S\right) - g_i\left(T\right)}{\varepsilon}.
			\end{align*}
			Therefore, we can write for $g_1g_2$,
			\begin{dgroup*}
				\begin{dmath*}
					D_{g_1g_2,T} = \lim\limits_{\varepsilon \rightarrow 0} \dfrac{g_1g_2\left(T+\varepsilon S\right) - g_1g_2\left(T\right)}{\varepsilon}
				\end{dmath*}
				\begin{dmath*}
					= \lim\limits_{\varepsilon \rightarrow 0} \dfrac{\left(g_1\left(T+\varepsilon S\right) - g_1\left(T\right)\right)g_2\left(T+\varepsilon S\right)}{\varepsilon} + \lim\limits_{\varepsilon \rightarrow 0} \dfrac{ g_1\left(T\right)\left(g_2\left(T+\varepsilon S\right) -  g_2\left(T\right)\right) }{\varepsilon}
				\end{dmath*}
				\begin{dmath*}
					= D_{g_1,T}g_2+g_1D_{g_2,T},
				\end{dmath*}
			\end{dgroup*}
			which shows that $g_1g_2\in \mathcal{C}^1_{op}(I)$. Since $I$ is arbitrary, the desired result follows immediately by applying the claim to $\mathbf{f}$ and $\mathbf{h}$.
		\end{proof}

		Since the system's observables belong to the commutant of the commutative von Neuman algebra $\mathscr{Y}_t$, $\mathscr{Y}_t'$, it is reasonable to approximate each element of $\mathbf{x}_t$, $x_{i,t}$, with a set of commutative estimates $\hat{x}_{i,t} \in \mathscr{Y}_t$. The difference $\tilde{x}_{i,t} = {x}_{i,t} - \hat{x}_{i,t}$ also belongs to the commutant algebra $\mathscr{Y}_t'$. Then for each $\tilde{x}_{i,t}$, both $\tilde{x}_{i,t}$ and $\mathscr{Y}_t$ generate a larger commutative algebra $\mathscr{Z}_{i,t}$ for which the Proposition \ref{prp:OrdinaryDerivative} can be invoked to infer the partial derivative of $\mathbf{f}\left(\mathbf{x}_t\right)$ with respect to $\mathbf{x}_t$ at $\hat{\mathbf{x}}_t$. We could not go further beyond the first order term of the Taylor series as in the classical nonlinear filter since the second order partial derivative of $\mathbf{f}\left(\mathbf{x}_t\right)$ will generally involve multiplication of two elements of $\tilde{x}_{i,t} \tilde{x}_{j,t}\in \mathscr{Y}_t'$ which generally do not commute with each other, and consequently a larger commutative algebra $\mathscr{Z}_{ij,t}$ generally does not exist.
		\\
		Our approximation begins by conjecturing that it is possible to construct a filter algorithm such that if the estimate $\hat{\mathbf{x}}_t$ is initially in the measurement algebra, it will always be inside it in the future,   $\hat{\mathbf{x}}_0 \in \mathscr{Y}_0 \Rightarrow \hat{\mathbf{x}}_t \in \mathscr{Y}_t$  $\forall t\geq 0$. In contrast to the formulation of quantum Kalman filter in\cite{wiseman2005,yamamoto2006robust}, since we will neglect the residual terms of Taylor series in our filter structure, $\hat{\mathbf{x}}_t$ is no longer optimal in the sense of the distance from $\mathbf{x}_t$  to a projection on $\mathscr{Y}_t$, i.e $\hat{\mathbf{x}}_t \neq \condExpect{\mathbf{x}_t} $. Nonetheless, each element of $\hat{\mathbf{x}}_t$ is commutative with respect to the other elements, as they are belong to the same commutative von Neumann algebra $\mathscr{Y}_t$.
		\\
		Since we require that $f_i \in  \mathcal{C}^1_{op}(\mathbb{R})$ and $\hat{\mathbf{x}}_t \in \mathscr{Y}_t$,  $\tilde{\mathbf{x}}_t \in \mathscr{Y}_t'$, 
		the condition in Proposition \ref{prp:OrdinaryDerivative} is satisfied. Consequently, we can write
		\begin{dgroup}
			\begin{dmath}
				\mathbf{f}(\mathbf{x}_t) = \mathbf{f}(\hat{\mathbf{x}}_t) + \left.\dfrac{\partial \mathbf{f}(\hat{\mathbf{x}}_t)}{\partial \mathbf{x}_t}\right|_{\mathbf{x}_t = \hat{\mathbf{x}}_t} \tilde{\mathbf{x}}_t + \mathbf{r}_f\left(\hat{\mathbf{x}_t}, \tilde{\mathbf{x}}_t\right), \label{eq:TaylorF}
			\end{dmath}
			\begin{dsuspend}
				with
			\end{dsuspend}
			\begin{dmath}
				\mathbf{F}(\hat{\mathbf{x}}_t) \equiv  \left.\dfrac{\partial \mathbf{f}(\mathbf{x}_t)}{\partial \mathbf{x}_t}\right|_{\mathbf{x}_t = \hat{\mathbf{x}}_t},
			\end{dmath}
			\begin{dsuspend}
				and
			\end{dsuspend}
			\begin{dmath}
				\mathbf{h}(\mathbf{x}_t) =  \mathbf{h}(\hat{\mathbf{x}}_t) + \left.\dfrac{\partial \mathbf{h}(\mathbf{x}_t)}{\partial \mathbf{x}_t}\right|_{\mathbf{x}_t = \hat{\mathbf{x}}_t} \tilde{\mathbf{x}}_t + \mathbf{r}_h\left(\hat{\mathbf{x}}_t, \tilde{\mathbf{x}}_t\right),
				\label{eq:TaylorH}
			\end{dmath}
			\begin{dmath}
				\mathbf{H}(\hat{\mathbf{x}}_t) \equiv \left.\dfrac{\partial \mathbf{h}(\mathbf{x}_t)}{\partial \mathbf{x}_t}\right|_{\mathbf{x}_t = \hat{\mathbf{x}}_t},
			\end{dmath}\label{eq:TaylorExpansion}
		\end{dgroup}
		where $\mathbf{r}_{(\cdot)}$ is the residual term of the Taylor expansion.
		The measurement error operator is given by
		\begin{dmath}
			d\mathbf{y}_t - d\hat{\mathbf{y}}_t = \left(\mathbf{H}(\hat{\mathbf{x}}_t) \tilde{\mathbf{x}}_t + \mathbf{r}_h\left(\hat{\mathbf{x}}_t, \tilde{\mathbf{x}}_t\right)  \right)dt + \mathbf{L}(\mathbf{x}_t) d\mathbf{A}_t^\ast + \mathbf{L}(\mathbf{x}_t)^\ast d\mathbf{A}_t + \mathbf{N}_t d\bm{\alpha}_t.
		\end{dmath}
		To construct a quantum EKF, we define a matrix of operators $\mathbf{P}_t \equiv \frac{1}{2}\condExpect{\anticommutator{\tilde{\mathbf{x}}_t}{\tilde{\mathbf{x}}_t}}$, that is the Hermitian variance of the estimation error.
		\\
		The quantum EKF problem is then given as follows. For a given open quantum system subjected to measurements with corresponding QSDEs given in \eqref{eq:QSDE} and \eqref{eq:MeasurementsSimplified} respectively, the following conditions are assumed :
		\begin{enumerate}
			\item Variances and cross correlation matrices $\mathbf{Q}_t,\mathbf{R}_t,\mathbf{S}_t$ are single valued.
			\item $\mathbf{R}_t$ is invertible.
			\item Initially $\hat{\mathbf{x}}_0 \in \mathscr{Y}_0$.
		\end{enumerate}
		We find a matrix $\mathbf{K}_t \in \mathscr{Y}_t$ corresponding to \eqref{eq:QuantumEKF} below, such that  $\hat{\mathbf{x}}_t \in \mathscr{Y}_t, \forall t\geq 0$ and $\mathbf{P}_t$ evolves according to the following Riccati differential equation upon neglecting the residual term of the Taylor series,
		\begin{dmath}
			d\hat{\mathbf{x}}_t = \mathbf{f}(\hat{\mathbf{x}}_t) dt + \mathbf{K}_t \left(d\mathbf{y}_t - d\hat{\mathbf{y}}_t\right), \label{eq:QuantumEKF}
		\end{dmath}
		\begin{dmath}
			\dfrac{d\mathbf{P}_t}{dt} =  \mathbf{F}(\hat{\mathbf{x}}_t)\mathbf{P}_t + \mathbf{P}_t\mathbf{F}(\hat{\mathbf{x}}_t)^\top + \mathbf{Q}_t  -  \left[\mathbf{P}_t\mathbf{H}(\hat{\mathbf{x}}_t)^\top + \mathbf{S}_t\right]\mathbf{R}_t^{-1}  \left[\mathbf{P}_t\mathbf{H}(\hat{\mathbf{x}}_t)^\top + \mathbf{S}_t\right]^\top.
			\label{eq:Riccati}
		\end{dmath}
		 We notice that since we neglect the non linear residual term in \eqref{eq:Riccati}, $\mathbf{P}_t$ is no longer interpreted as the variance of the estimation error, but rather a general positive definite matrix of operators which will be involved in the dynamics of the filter. Without loss of the generality, we can assume that $\mathbf{P}_0 \in \mathscr{Y}_0$. The following theorem shows the existence of a quantum EKF satisfying the condition above.
		\begin{theorem}
			Consider an open quantum system described by the QSDEs given in  \eqref{eq:QSDE} subjects to the measurements given in  \eqref{eq:Measurements}. Then, there exists a Kalman gain $\mathbf{K}_t \in \mathscr{Y}_t$,
			\begin{align}
			\mathbf{K}_t =& \left[\mathbf{P}_t\mathbf{H}(\hat{\mathbf{x}}_t)^\top + \mathbf{S}_t\right]\mathbf{R}_t^{-1} , \label{eq:KalmanGain}
			\end{align}
			such that if the quantum extended Kalman filter is given by \eqref{eq:QuantumEKF},	then $\hat{\mathbf{x}}_t \in \mathscr{Y}_t, \forall t\geq 0$ and $\mathbf{P}_t$ evolves according to \eqref{eq:Riccati} upon neglecting the residual term of the Taylor series in \eqref{eq:TaylorExpansion}.
		\end{theorem}
		\begin{proof}
			To establish the first part of the theorem, we need to show that there exists $\mathbf{K}_t \in \mathscr{Y}_t$, such that $\hat{\mathbf{x}}_t \in \mathscr{Y}_t, \forall t\geq 0$.
			The condition that $\hat{\mathbf{x}}_t \in \mathscr{Y}_t, \forall t \geq 0$ follows from the causality of the quantum EKF given in  \eqref{eq:QuantumEKF}. But to make it clear, let $\mathbf{K}_t \in \mathscr{Y}_t$. We observe that the differential equation of the quantum EKF can be written in integral form as below
			\begin{dgroup*}
				\begin{dmath*}
					\int_{0}^{t} d\hat{\mathbf{x}}_s = \int_{0}^{t}\left[\mathbf{f}(\hat{\mathbf{x}}_s) - \mathbf{K}_s \mathbf{h}\left(\hat{\mathbf{x}}_s\right)\right]ds +  \int_{0}^{t} \mathbf{K}_s d\mathbf{y}_s,
				\end{dmath*}

			\end{dgroup*}
			where the first term of the right hand side involves Reimann-Stieltjes integration. Since the integral above is defined then, there exists a partition of times $0 = t_0 \leq t_1 \leq \dots t_N = t$, $\Delta t_i \equiv t_i - t_{i-1}$,
			such that we can define an infinite sum, by taking $ \lim\limits_{N\rightarrow \infty}\sup_{i} \Delta t_i = 0$, regarding the second integration in $d\mathbf{y}_s$ in the It\^{o} sense,
			\begin{dgroup*}
				\begin{dmath*}
					\hat{\mathbf{x}}_t =  \hat{\mathbf{x}}_0 + \lim\limits_{N\rightarrow \infty}  \sum_{i=0}^{N-1} \left[\mathbf{f}(\hat{\mathbf{x}}_{t_i}) - \mathbf{K}_{t_i} \mathbf{h}\left(\hat{\mathbf{x}}_{t_i}\right)\right]  \Delta t_{i+1} + \lim\limits_{N\rightarrow \infty}  \sum_{i=0}^{N-1} \mathbf{K}_{t_i} \left(\mathbf{y}_{t_{i+1}} - \mathbf{y}_{t_i}\right).
				\end{dmath*}
			\end{dgroup*}
			Hence, it is clear that if $\hat{\mathbf{x}}_0,\in \mathscr{Y}_0$, then $\hat{\mathbf{x}}_t \in \mathscr{Y}_t, \forall t\geq 0$ , since it is in the span of $\mathbf{y}_{t_i}$ and $\hat{\mathbf{x}}_{t_i}$ $\forall t_i < t$.
			\\
			For the second part of the theorem, let us now expand $d\hat{\mathbf{x}}_t$ according to the Taylor expansion in  \eqref{eq:TaylorF} and  \eqref{eq:TaylorH}. Doing that leads us to the following equation,
			\begin{dmath*}
				d\hat{\mathbf{x}}_t = \mathbf{f}(\hat{\mathbf{x}}_t) dt + \mathbf{K}_t \left(\left(\mathbf{H}(\hat{\mathbf{x}}_t) \tilde{\mathbf{x}}_t + \mathbf{r}_h\left(\hat{\mathbf{x}}_t, \tilde{\mathbf{x}}_t\right)  \right)dt + \mathbf{L}(\mathbf{x}_t)  d\mathbf{A}_t^\ast + \mathbf{L}(\mathbf{x}_t)^\ast d\mathbf{A}_t + \mathbf{N}_t d\bm{\alpha}_t\right).
			\end{dmath*}
			From the equation above, the estimation error can be given by the following equations,
			\begin{dgroup}

				\begin{dmath}
					d\tilde{\mathbf{x}}_t = \left(\mathbf{F}(\hat{\mathbf{x}}_t) - \mathbf{K}_t\mathbf{H}(\hat{\mathbf{x}}_t)\right)\tilde{\mathbf{x}}_t dt + \mathbf{r} dt + \left(\mathbf{G}(\mathbf{x}_t)^\ast - \mathbf{K}_t \mathbf{L}(\mathbf{x}_t)^\ast\right)d\mathbf{A}_t + \left(\mathbf{G}(\mathbf{x}_t) - \mathbf{K}_t \mathbf{L}(\mathbf{x}_t) \right)d\mathbf{A}_t^\ast - \mathbf{K}_t \mathbf{N}_t d\bm{\alpha}_t ,\label{eq:errorDynamic}
				\end{dmath}
				\begin{dmath}
					\mathbf{r}_t\left(\hat{\mathbf{x}}_t,\tilde{\mathbf{x}}_t\right) =  \mathbf{r}_f\left(\hat{\mathbf{x}_t}, \tilde{\mathbf{x}}_t\right) -  \mathbf{K}_t\mathbf{r}_h\left(\hat{\mathbf{x}_t}, \tilde{\mathbf{x}}_t\right). \label{eq:residual}
				\end{dmath}
			\end{dgroup}
			Using the quantum It\^{o} multiplication rule\cite{parthasarathy2012}, we can obtain the estimation variance dynamics as follow
			\begin{dgroup}
				\begin{dmath}
					\dfrac{d\mathbf{P}_t}{dt} =  \left(\mathbf{F}(\hat{\mathbf{x}}_t) - \mathbf{K}_t\mathbf{H}(\hat{\mathbf{x}}_t)\right) \mathbf{P}_t+ \mathbf{P}_t\left(\mathbf{F}(\hat{\mathbf{x}}_t) - \mathbf{K}_t\mathbf{H}(\hat{\mathbf{x}}_t)\right)^\top + \bm{\Theta}(\hat{\mathbf{x}}_t,\tilde{\mathbf{x}}_t)
					+ \dfrac{\hbar}{2}  \condExpect{ \anticommutator{d\tilde{\mathbf{x}}_t}{d\tilde{\mathbf{x}}_t}}, \label{eq:RiccatiNotComplete}
				\end{dmath}
				\begin{dsuspend}
					with,
				\end{dsuspend}
				\begin{dmath}
					\bm{\Theta}(\hat{\mathbf{x}}_t,\tilde{\mathbf{x}}_t) \equiv \dfrac{1}{2} \condExpect{\anticommutator{\mathbf{r}_t\left(\hat{\mathbf{x}}_t,\tilde{\mathbf{x}}_t\right)}{\tilde{\mathbf{x}}_t}}.
				\end{dmath}

			\end{dgroup}
			Now, from  \eqref{eq:RiccatiNotComplete}, by the definition of variances given in  \eqref{eq:Q},\eqref{eq:R} and \eqref{eq:S}, and that $\mathbf{K}_t \in \mathscr{Y}_t$, we obtain
			\begin{dmath}
				\dfrac{d\mathbf{P}_t}{dt} =  \mathbf{F}(\hat{\mathbf{x}}_t)\mathbf{P}_t + \mathbf{P}_t\mathbf{F}(\hat{\mathbf{x}}_t)^\top + \mathbf{Q}_t + \mathbf{K}_t \mathbf{R}_t {\mathbf{K}_t}^\top - \left(\mathbf{K}_t \left(\mathbf{H}(\hat{\mathbf{x}}_t)\mathbf{P}_t + \mathbf{S}_t^\top\right) + {\left(\mathbf{H}(\hat{\mathbf{x}}_t)\mathbf{P}_t + \mathbf{S}_t^\top\right)}^\top \mathbf{K}_t^\top \right) + \bm{\Theta}(\hat{\mathbf{x}}_t,\tilde{\mathbf{x}}_t).
				\label{eq:RiccatiNotOptimal}
			\end{dmath}
			Since $\mathbf{P}_t$,$\mathbf{R}_t$,$\mathbf{S}_t$, and $\hat{\mathbf{x}}_t$ belong to $\mathscr{Y}_t$, then the Kalman gain given in \eqref{eq:KalmanGain} also belongs to $\mathscr{Y}_t$. Substituting \eqref{eq:KalmanGain} to \eqref{eq:RiccatiNotOptimal}, and ignoring the nonlinearity $\bm{\Theta}$ we obtain the desired result as in \eqref{eq:Riccati}.
		\end{proof}
		Before we go further, we would like to address the implementability of the quantum EKF in \eqref{eq:QuantumEKF}. In practical applications, we would often be given initial values of $\hat{\mathbf{x}}_t$ and $\mathbf{P}_t$, rather than a complete description of a set of operators in an underlying Hilbert space. Furthermore, in many cases, the interest is only to estimate the mean value and covariance of $\mathbf{x}_t$.
		\\
		Given such conditions, we would like to see the relation between the quantum EKF and the classical EKF. From  \eqref{eq:QuantumEKF}, we have the evolution of $\hat{\mathbf{x}}_t \in \mathscr{Y}_t$, which can be written as
		\begin{align}
		d\hat{\mathbf{x}}_t =& \left[\mathbf{f}(\hat{\mathbf{x}}_t) - \mathbf{K}_t \mathbf{h}\left(\hat{\mathbf{x}}_t\right) \right] dt + \mathbf{K}_t d\mathbf{y}_t. \label{eq:QuantumEKF2}
		\end{align}
		By definition $d\mathbf{y}_t \in \mathscr{Y}_t$ and hence by Theorem \ref{thm:SpectralTheoremInfiniteDimension} there exist a $\ast-$ isomorphism $\tau$ from $\mathscr{Y}_t$ to $\mathcal{L}^\infty\left(\Omega,\mathscr{F},\mu\right)$. From now on, we write $\tau\left(\cdot\right)_{t,\omega} \in \mathbb{C}$ as $(\cdot)_{t,\omega}$.  Then we write $\forall \omega \in \Omega,\forall t \geq 0$,
		\begin{dmath}
			d\hat{\mathbf{x}}_{t,\omega} =\left[\mathbf{f}(\hat{\mathbf{x}}_{t,\omega}) - \mathbf{K}_{t,\omega} \mathbf{h} \left( \hat{\mathbf{x}}_{t,\omega}\right)  \right] dt +\mathbf{K}_{t,\omega} d\mathbf{y}_{t,\omega} \; \; ,
			\label{eq:classicalEKF}
		\end{dmath}
		which is an ordinary stochastic differential equation. Using the same  $\ast-$ isomorphism $\tau$, the dynamics of the estimation error variance, (see the Riccati equations in \eqref{eq:Riccati}) can also be transformed into classical Riccati differential equations. This transformation in turn makes the quantum EKF implementable as a recursive filter in a digital signal processor.
		\begin{remark}
			It is worth emphasizing that the essential difference between the quantum EKF and the classical EKF is the fact that the set of the system's observables $\mathbf{x}_t$ do not belong to a commutative von Neumann algebra. In addition, the dynamics of $\mathbf{x}_t$ generally consist of non-commutating operators and hence there is no $*-$ isomorphism that can transform $\mathbf{x}_t$ and its dynamics into a measurable function on a classical probability space $\left(\Omega,\mathscr{F},\mu\right)$. Otherwise, if this is the case, then the quantum EKF problem is reduced to the classical EKF problem.
		\end{remark}
		For photon counting measurements,  a Poisson processes can be written as a sum of two independent quantum Gaussian noises as in  \eqref{eq:CountingMeasurement}. This enables us to treat a filtering problem for both diffusive and jump measurements simultaneously. This is unique to quantum stochastic filtering, since in a classical settings, one can never have a transformation from a jump random process into a continuous processes, \cite{bouten2007introduction}.

		\subsection{One Step Computational Complexity of the SME and the quantum EKF} \label{sec:Complexity}
		Here we present a comparison of the computational complexity of the SME and the quantum EKF. The SME for $m$ output measurement channels is given by \cite{MuhammadF.Emzir2015},
		\begin{dmath}
			d\rho_t =\left[-i\left[\mathbb{H},\rho_t\right] + \mathbb{L}^{\top}\rho_t\mathbb{L}^{\ast} - \frac{1}{2}\mathbb{L}^{\dagger}\mathbb{L}\rho_t - \frac{1}{2}\rho_t\mathbb{L}^{\dagger}\mathbb{L}\right]dt + \zeta_{\rho}^{\top}\bm{\Gamma}^{-\top}d\mathbf{W}. \label{eq:SME}
		\end{dmath}
		In this equation, $\rho_t$ is the system's conditional density operator and $d\mathbf{W}$ is the error vector between the expected value and the measurement. The weighting function $\zeta_{\rho}^{\top}\bm{\Gamma}^{-\top}$ relates the contribution of each measurement to the total increment of the conditional density operator. Now, suppose we have $N_m$ subsystems, and under truncation let each system's Hilbert space dimension be $N_s$.  
		The computational complexity of \eqref{eq:SME}, will be of order $O\left(N_m (N_s^{N_m})^3 \right) = O\left(N_m N_s^{3N_m} \right)$. Now, suppose we want to estimate $n$ observables of the system. Then after propagating the SME, we need to calculate $\trace{\rho_t x_i}, i \leq n$, which is on the order of $O\left(n (N_s^{N_m})^3 \right)$. Consequently to propagate $n$ observables from each subsystem from SME, we will need a calculation effort  $ \sim \left(\left(\zeta_1 n+\zeta_2 N_m\right) N_s^{3 N_m} \right)$, for some $\zeta_1,\zeta_2 >1$.
		\\
		In contrast, after transforming the Riccati and quantum EKF equations to the standard SDE, the computational complexity of the quantum EKF is the same as that of the classical EKF. The EKF computational effort only depends on $n,m$ and $N_m$, and the complexity of evaluating  $\mathbf{f}$, and the Jacobian matrices $\mathbf{F}$ and $\mathbf{H}$ .
		In a single time step, one has to propagate the Riccati equation in \eqref{eq:Riccati}, which has the complexity  $O\left( n^{3 N_m} \right)$, calculation of the Jacobian matrices $\mathbf{F}$ and $\mathbf{H}$ which could vary depend on the type of the function involved, plus solving the quantum EKF \eqref{eq:QuantumEKF}. Evaluating \eqref{eq:QuantumEKF} involves the calculation of $\mathbf{f}$ which also can vary, and matrix-vector calculation in $\mathbf{K}_t \left(d\mathbf{y}-\hat{d\mathbf{y}}\right)$, which is $O(m^2 (n^{N_m}))$.

		\subsection{Convergence analysis}\label{sec:ConvAnal}
		Here we will establish a stochastic convergence condition for the quantum EKF. The approach we pursue here is closely related to the stochastic convergence analysis of the classical extended Kalman filter and deterministic nonlinear observer design \cite{Reif2000,reif1998ekf,reif1999nonlinear}.  In essence, the main difference between the classical EKF stochastic convergence proof and what we present here is the use of the semi-norm in \eqref{eq:SemiNorm} in the place of Euclidean norm, and the quantum Markovian process generator. Moreover, due to the coupling nature of the measurement and process noise in every open quantum system, we need to assume the positive definiteness of $\mathbf{Q}_t-\mathbf{S}_t\mathbf{R}_t^{-1}\mathbf{S}_t^\top$.
		The following assumptions, definitions, and lemmas will serve as a foundation for the local convergence condition of the quantum EKF estimation errors.

		\begin{assumption}
			With the definition of semi-norm given in \eqref{eq:SemiNorm}, together with our previous assumption that $\mathbf{R}_t$ is positive definite matrix of operators in $\mathscr{Y}_t$, we also need the following assumptions in our analysis.
			\begin{enumerate}[label={\bf A\Roman*}]

				\item  Let $\mathscr{E}_f,\mathscr{E}_h \subset \mathscr{Y}_t$. For $\mathscr{E}_f = \left\{\hat{\mathbf{x}}_t: \norm{\mathbf{x}_t - \hat{\mathbf{x}}_t} \leq \epsilon_f \right\}$, $\mathscr{E}_h  = \left\{\hat{\mathbf{x}}_t: \norm{\mathbf{x}_t - \hat{\mathbf{x}}_t} \leq \epsilon_h \right\}$, $\epsilon_f,\epsilon_h > 0$, $\exists r_h,r_f > 0$, such that the residual term of $\mathbf{r}_f$ and $\mathbf{r}_h$ satisfy the following rate constraint
				\begin{align*}
				\norm{\mathbf{r}_f}  < r_f \norm{\tilde{\mathbf{x}}_t}^2, \forall \hat{\mathbf{x}}_t \in  \mathscr{E}_f, \\
				\norm{\mathbf{r}_h}  < r_h \norm{\tilde{\mathbf{x}}_t}^2 , \forall \hat{\mathbf{x}}_t \in  \mathscr{E}_h,
				\end{align*}\label{asm:ResidualAsm}

				\item The operator valued matrix $\mathbf{H}$, and the cross correlation $\mathbf{S}$ are bounded from above,
				\begin{align}
				\norm{\mathbf{H}\left(\mathbf{x}\right)} \leq \bar{h},\\
				\norm{\mathbf{S}} \leq \bar{s}.
				\end{align}

				This assumption follows from the fact that we restrict our observable to the von Neumann algebra $\mathscr{N} = \mathcal{B}$. Moreover, since $\mathbf{h} \in \mathcal{C}^1_{op} \subset \mathcal{C}^1$, $\mathbf{H}$ is bounded.\label{asm:Hbounded}

				\item  $\mathbf{P}_t$ is always greater than zero and bounded,  i.e,
				\begin{align}
				0 < \underline{p}\mathbf{I} \leq \mathbf{P}_t \leq \bar{p}\mathbf{I} \;, \forall t\geq 0.
				\end{align}
				If we consider $\mathbf{P}_t$ as an estimation error covariance, this condition will generally be valid in quantum mechanical system estimation since the Heisenberg inequality dictates that $0 < \underline{p}\mathbf{I}$. Furthermore, it seems also natural to consider $\mathbf{P}_t$ to be bounded from above.\label{asm:PAsm}

				\item
				\begin{align}
				\mathbf{Q}_t - \mathbf{S}_t\mathbf{R}_t^{-1}\mathbf{S}_t^\top \geq 0, \;\forall t \geq 0.
				\end{align}
				As we will encounter later on in Lemma \ref{lem:Dissipativity}, to satisfy the dissipative inequality, we require that $\mathbf{Q}_t-\mathbf{S}_t\mathbf{R}_t^{-1}\mathbf{S}_t^\top > 0$ is always satisfied. However, as we will show in the Proposition \ref{prp:Q_SRS_inequality}, we can only obtain a sufficient condition that $m \mathbf{Q}_t-\mathbf{S}_t\mathbf{R}_t^{-1}\mathbf{S}_t^\top \geq 0$. In the next section, we will show how to deal with this restriction.\label{asm:QRSAsm}
			\end{enumerate}
			\label{asm:Assumptions}
		\end{assumption}

		We will now show the inequality,  $m \mathbf{Q}_t-\mathbf{S}_t\mathbf{R}_t^{-1}\mathbf{S}_t^\top \geq 0$, by using the following Lemma,
		\begin{lemma}\label{lem:SelfNonDemolition}
			If a set of $m$ measurement $\mathbf{y}_t$ satisfying self-non demolition property, then $\commutator{d y_{i,t}}{ d y_{j,t}} = 0$ $\forall i,j \leq m, t\geq 0$.
		\end{lemma}
		\begin{proof}
			To see this, we first examine that if $\mathbf{y}_t$ satisfying self-non demolition property, we could write for $i,j \leq m$
			\begin{dgroup*}
				\begin{dmath*}
					d \commutator{y_{i,t}}{ y_{j,t}} = 0 = \commutator{d y_{i,t}}{ y_{j,t}} +\commutator{y_{i,t}}{d y_{j,t}}
					+\commutator{d y_{i,t}}{d y_{j,t}}
				\end{dmath*},
			\end{dgroup*}
			but
			$\commutator{d y_{i,t}}{y_{j,t}} = \commutator{y_{i,t}}{d y_{j,t}} = 0$. This can be seen from $d y_{i,t}  = \sum_{k} U_t^\ast\left(z_k\otimes \mathbb{1}\right)U_t db_k$,
			where $db_k \in \left\{d\mathbf{A}_i,d\mathbf{A}^\ast_i,d\bm{\alpha}_i,dt\right\}, i\leq m$,
			and $z \in \mathcal{B}\left(\mathscr{H}_s\right)$, whilst $y_{j,t}  = U_t^\ast\left(\mathbb{1} \otimes y \right)U_t $, where $y \in \mathcal{B}\left(\Gamma(\mathscr{h})_t\}\right)$. Hence $\commutator{d y_{i,t}}{d y_{j,t}} = 0$.
		\end{proof}

		\begin{proposition} \label{prp:Q_SRS_inequality}
			For an open quantum system with QSDEs and measurement given in \eqref{eq:QSDE} and \eqref{eq:Measurements} respectively, the covariance matrices $\mathbf{Q}_t,\mathbf{R}_t,\mathbf{S}_t$ satisfy the following inequality
			\begin{align}
			m \mathbf{Q}_t-\mathbf{S}_t\mathbf{R}_t^{-1}\mathbf{S}_t^\top \geq 0, \; \forall t \geq 0.
			\end{align}
			\begin{proof}
				Let $x_{i,t}$ and $y_{j,t}$ denote the $i$ and $j$ th elements of $\mathbf{x}_t$ and $\mathbf{y}_t$ respectively.  From Lemma \ref{lem:SelfNonDemolition}, we have $\commutator{d y_{i,t}}{ d y_{j,t}} = 0$ $\forall i,j \leq m, t\geq 0$.
				\\
				Let $\bar{\mathbf{S}}_t = d\mathbf{x}_t d\mathbf{y}_t^\top$, $\underline{\mathbf{S}}_t  = d\mathbf{y}_t d\mathbf{x}_t^\top$, $\bar{\mathbf{Q}}_t = \underline{\mathbf{Q}}_t = d\mathbf{x}_t d\mathbf{x}_t^\top$	, and  $\bar{\mathbf{R}}_t = \underline{\mathbf{R}}_t^\top = d\mathbf{y}_t d\mathbf{y}_t^\top$.
				\\
				We first claim that there exists a symmetric matrix of operators $\mathbf{M}_t \in \mathscr{Y}_t' $ such that $\anticommutator{d\mathbf{x}_t}{d\mathbf{x}_t}$ as below,
				\begin{dmath}
					\anticommutator{d\mathbf{x}_t}{d\mathbf{x}_t} = \bar{\mathbf{S}}_t\mathbf{M}_t\underline{\mathbf{S}}_t + \left(\bar{\mathbf{S}}_t\mathbf{M}_t\underline{\mathbf{S}}_t\right)^\top. \label{eq:SMS}
				\end{dmath}
				To see this, we can write
				\begin{dmath*}
					\anticommutator{d\mathbf{x}_t}{d\mathbf{x}_t}_{i,j}
					= dx_{i,t} \left[\sum_{k,k' = 1}^{m,m} dy_{k,t}\mathbf{M}_{k,k'} dy_{k',t}\right]dx_{j,t} + dx_{j,t} \left[\sum_{k,k' = 1}^{m,m} dy_{k',t}\mathbf{M}_{k,k'} dy_{k,t} \right]dx_{i,t},
				\end{dmath*}
				by requiring,
				\begin{dmath}
					\sum\limits_{i,j = 1}^{m,m} dy_{i,t} M_{i,j} dy_{j,t} = \mathbb{1}. \label{eq:mElement}
				\end{dmath}
				From  \eqref{eq:mElement}, since  $\mathbf{M}_t \in \mathscr{Y}_t' $, we have,
				\begin{dmath*}
					{\sum\limits_{i,j = 1}^{m,m} dy_{i,t} M_{i,j} dy_{j,t} = \mathbb{1} = \sum\limits_{i,j = 1}^{m,m} dy_{i,t} dy_{j,t} M_{i,j}} ={ \sum\limits_{i,j = 1}^{m,m} dy_{i,t} dy_{j,t} M_{j,i} = \trace{\bar{\mathbf{R}}_t \mathbf{M}_t}}.
				\end{dmath*}
				Hence, selecting $\mathbf{M}_t = \frac{1}{m} \bar{\mathbf{R}}_t^{-1} \in \mathscr{Y}_t \subset \mathscr{Y}_t'$ the assertion is verified. Now since $\mathbf{S}_t\mathbf{R}_t^{-1}\mathbf{S}_t^\top$ is a convex function with respect to both $\mathbf{S}_t$ and $\mathbf{R}_t$ (see \cite[Proposition 8.6.17 (xiv)]{bernstein2009matrix}), by the Jensen inequality
				\begin{dmath*}
					\mathbf{S}_t\mathbf{R}_t^{-1}\mathbf{S}_t^\top = \frac{1}{dt}  \condExpect{\frac{1}{2}\left(\bar{\mathbf{S}}_t+\underline{\mathbf{S}}_t^\top\right)} \condExpect{\bar{\mathbf{R}}_t}^{-1} \condExpect{\frac{1}{2}\left(\bar{\mathbf{S}}_t^\top+\underline{\mathbf{S}}_t\right)}\leq
					\frac{1}{dt}\frac{1}{4}\condExpect{\left(\bar{\mathbf{S}}_t+\underline{\mathbf{S}}_t^\top\right) \bar{\mathbf{R}}_t^{-1}  \left(\bar{\mathbf{S}}_t^\top+\underline{\mathbf{S}}_t\right)}
					\leq \frac{1}{dt}\frac{1}{2}\condExpect {\bar{\mathbf{S}}_t \bar{\mathbf{R}}_t^{-1} \underline{\mathbf{S}}_t + \left(\bar{\mathbf{S}}_t \bar{\mathbf{R}}_t^{-1} \underline{\mathbf{S}}_t\right)^\top}
					= \frac{m}{dt}\frac{1}{2}\condExpect {\bar{\mathbf{S}}_t \mathbf{M}_t \underline{\mathbf{S}}_t + \left(\bar{\mathbf{S}}_t \mathbf{M}_t \underline{\mathbf{S}}_t\right)^\top}
					=
					\frac{m}{dt}\condExpect{\frac{1}{2}\left(\bar{\mathbf{Q}}_t+\underline{\mathbf{Q}}_t^\top\right)} = m \mathbf{Q},
				\end{dmath*}
				which completes the proof.
			\end{proof}
		\end{proposition}
		The following Lemma gives a bound on the nonlinear residual rate based on Assumption \ref{asm:Assumptions}, see also \cite[Lemma 3.3]{Reif2000} for an analogous of result for the classical EKF.

		\begin{lemma}\label{lem:NonlinearityBound}
			Consider a functional $\phi$ that is a function of the residual $\mathbf{r}_t$ and the estimation error $\tilde{\mathbf{x}}_t$ given below,
			\begin{align}
			\phi(\mathbf{r}_t,\tilde{\mathbf{x}}_t) \equiv & \mathbf{r}_t^\top \mathbf{P}_t^{-1} \tilde{\mathbf{x}}_t + \tilde{\mathbf{x}}_t^\top  \mathbf{P}_t^{-1} \mathbf{r}_t.
			\end{align}
			Then under Assumptions \ref{asm:Assumptions}, there exists positive $\epsilon$ and $\kappa$ such that for every $\norm{\tilde{\mathbf{x}}_t} \leq \epsilon$, then
			\begin{dmath}
				\norm{\phi(\mathbf{r}_t,\tilde{\mathbf{x}}_t)} \leq \kappa \norm{\tilde{\mathbf{x}}_t}^3. \label{eq:phi}
			\end{dmath}
		\end{lemma}
		\begin{proof}
			From  \eqref{eq:residual}, we have
			\begin{dmath*}
				\phi(\mathbf{r}_t,\tilde{\mathbf{x}}_t) = \mathbf{r}_f^\top\mathbf{P}_t^{-1}\tilde{\mathbf{x}}_t + \tilde{\mathbf{x}}_t^\top\mathbf{P}_t^{-1}\mathbf{r}_f - \left[\mathbf{r}_h^\top\mathbf{K}_t^\top\mathbf{P}_t^{-1}\tilde{\mathbf{x}}_t + \tilde{\mathbf{x}}_t^\top\mathbf{P}_t^{-1}\mathbf{K}_t\mathbf{r}_h\right].
			\end{dmath*}
			Furthermore using \ref{asm:PAsm} and \ref{asm:Hbounded},
			\begin{align*}
			\bar{k} \equiv \norm{\mathbf{K}_t} \leq \dfrac{\left(\bar{p}\bar{h}+\bar{s}\right)}{\underline{r}}.
			\end{align*}
			From \ref{asm:ResidualAsm},  taking $\epsilon = \min({\epsilon_f},{\epsilon_h})$
			\begin{align*}
			\norm{\phi(\mathbf{r}_t,\tilde{\mathbf{x}}_t)} \leq & 2\left[\dfrac{r_h \bar{k}+r_f}{\underline{p}}\right]\norm{\tilde{\mathbf{x}}_t}^3 \equiv \kappa \norm{\tilde{\mathbf{x}}_t}^3.
			\end{align*}
		\end{proof}

		To examine the convergence behavior of the coupling dynamics between the estimation error and the covariance, we consider a Lyapunov positive operator $V \in \mathcal{O}_+\left(\mathscr{H}_{\mathsf{T}}\right)$ given by
		\begin{align}
		V\left(\tilde{\mathbf{x}}_t\right) = \tilde{\mathbf{x}}_t^\top \mathbf{P}_t^{-1} \tilde{\mathbf{x}}_t. \label{eq:Lyanpunov}
		\end{align}

		\begin{definition}[Class $\mathcal{K}$, and class $\mathcal{KR}$ function]\cite{sastry1999nonlinear} A function $\alpha( \cdot ) : \mathbb{R}_+ \rightarrow \mathbb{R}_+$  belongs to class $\mathcal{K}$ if it is continuous and strictly increasing, where $\alpha(0) = 0$. A function $\alpha( \cdot )$ is said to belong class $\mathcal{KR}$  if $\alpha$ is of class $\mathcal{K}$ and in addition, $\lim\limits_{z \rightarrow \infty}\alpha(z) = \infty$.
		\end{definition}

		Notice that, under \ref{asm:PAsm}, the Lyapunov candidate functional $V$ is positive definite and \emph{decresent}. That is we can find $\alpha(\norm{\mathbf{x}}),\beta(\norm{\mathbf{x}})$ of class $\mathcal{K}$, such that
		\begin{align*}
		\alpha(\norm{\tilde{\mathbf{x}}_t}) \leq \Paverage{V(\tilde{\mathbf{x}}_t)} \leq \beta(\norm{\tilde{\mathbf{x}}_t}).
		\end{align*}
		The following Lemma shows the dissipativity of the quantum EKF estimation error.
		\begin{lemma}\label{lem:Dissipativity}
			Consider an open quantum system QSDE given in \eqref{eq:QSDE}, and measurement given in \eqref{eq:Measurements}. With the quantum EKF given in \eqref{eq:QuantumEKF}, and $\kappa,\epsilon \geq 0 $ in Lemma \ref{lem:NonlinearityBound}. Under Assumptions \ref{asm:Assumptions}, there exist $\epsilon' >0$ such that if $\norm{\tilde{\mathbf{x}}_t} \leq \epsilon', \forall t \geq 0$,
			there exist $\gamma,\delta >0$ such that the Lyapunov positive operator \eqref{eq:Lyanpunov} satisfies the dissipation inequality below
			\begin{align}
			\mathcal{L}\left(V\right)  \leq & -\gamma V(\tilde{\mathbf{x}}_t)  + \delta. \label{eq:DissipativityDifferential}
			\end{align}

		\end{lemma}
		\begin{proof}
			We could write the time derivative of $\mathbf{P}_t^{-1}$ as below.
			\begin{align*}
			\dfrac{\partial \mathbf{P}_t^{-1}}{\partial t} =& - \mathbf{P}_t^{-1}\dfrac{\partial \mathbf{P}_t}{\partial t}\mathbf{P}_t^{-1}.
			\end{align*}
			Then by It\^{o} expansion, we have
			\begin{dmath*}
				dV =   \tilde{\mathbf{x}}_t^\top\mathbf{P}_t^{-1}d\tilde{\mathbf{x}}_t+d\tilde{\mathbf{x}}_t^\top\mathbf{P}_t^{-1}\tilde{\mathbf{x}}_t+d\tilde{\mathbf{x}}_t^\top\mathbf{P}_t^{-1}d\tilde{\mathbf{x}}_t - \tilde{\mathbf{x}}_t^\top\mathbf{P}_t^{-1}\dfrac{\partial \mathbf{P}_t}{\partial t}\mathbf{P}_t^{-1}\tilde{\mathbf{x}}_t dt.
			\end{dmath*}
			The quantum Markov generator for the Lyapunov positive operator $V$ is then given by
			\begin{dmath*}
				\mathcal{L}\left(V\right)  =  \tilde{\mathbf{x}}_t^\top\mathbf{P}_t^{-1}\left[\left(\mathbf{F}(\hat{\mathbf{x}}_t) - \mathbf{K}_t\mathbf{H}(\hat{\mathbf{x}}_t)\right)\tilde{\mathbf{x}}_t  + \mathbf{r} \right]+\left[\left(\mathbf{F}(\hat{\mathbf{x}}_t) - \mathbf{K}_t\mathbf{H}(\hat{\mathbf{x}}_t)\right)\tilde{\mathbf{x}}_t  + \mathbf{r} \right]^\top\mathbf{P}_t^{-1}\tilde{\mathbf{x}}_t
				+  \trace{\left(\mathbf{G}(\mathbf{x}_t)^\ast - \mathbf{K}_t\mathbf{L}(\mathbf{x}_t)^\ast\right)^\top\mathbf{P}_t^{-1}\left(\mathbf{G}(\mathbf{x}_t) - \mathbf{K}_t \mathbf{L}(\mathbf{x}_t)\right)} - \tilde{\mathbf{x}}_t^\top\mathbf{P}_t^{-1}\dfrac{\partial \mathbf{P}_t}{\partial t}\mathbf{P}_t^{-1}\tilde{\mathbf{x}}_t.
			\end{dmath*}
			From Lemma \ref{lem:NonlinearityBound},  \eqref{eq:errorDynamic}, \eqref{eq:KalmanGain}, and \eqref{eq:Riccati}, there exists $\epsilon >0$ such that,
			\begin{dmath*}
				\mathcal{L}\left(V\right) \leq   -\tilde{\mathbf{x}}_t^\top\left[ \mathbf{P}_t^{-1}\left[\mathbf{Q}_t + \mathbf{P}\mathbf{H}_t^\top\mathbf{R}_t^{-1}\mathbf{H}_t\mathbf{P} - \left[\mathbf{S}_t\mathbf{R}_t^{-1}\mathbf{S}_t^\top + \kappa  \norm{\tilde{\mathbf{x}}_t} \mathbf{P}^2_t\right]\right]\mathbf{P}_t^{-1}\right]\tilde{\mathbf{x}}_t + \delta
				\leq  -\tilde{\mathbf{x}}_t^\top\left[ \mathbf{P}_t^{-1}\left[\left(\mathbf{Q}_t-\mathbf{S}_t\mathbf{R}_t^{-1}\mathbf{S}_t^\top\right) + \mathbf{P}_t\left[\mathbf{H}_t^\top\mathbf{R}_t^{-1}\mathbf{H}_t - \kappa  \norm{\tilde{\mathbf{x}}_t} \right]\mathbf{P}_t\right]\mathbf{P}_t^{-1}\right]\tilde{\mathbf{x}}_t + \delta .
			\end{dmath*}
			By assumption \ref{asm:QRSAsm}, $\mathbf{Q}_t-\mathbf{S}_t\mathbf{R}_t^{-1}\mathbf{S}_t^\top > 0$. Now select $\epsilon'$, such that, $\forall \norm{\tilde{\mathbf{x}}_t} \leq \epsilon'$
			\begin{dmath*}
				\norm{\left(\mathbf{Q}_t-\mathbf{S}_t\mathbf{R}_t^{-1}\mathbf{S}_t^\top\right) + \mathbf{P}_t\left[\mathbf{H}_t^\top\mathbf{R}_t^{-1}\mathbf{H}_t - \kappa  \norm{\tilde{\mathbf{x}}_t} \right]\mathbf{P}_t} \geq 0.
			\end{dmath*}
			This fulfilled by taking,
			\begin{dmath*}
				{\epsilon' = \min\left(\dfrac{q_e}{\bar{p}^2 \kappa},\epsilon\right)},
			\end{dmath*}
			where $q_e  = \inf{\norm{\mathbf{Q}_t-\mathbf{S}_t\mathbf{R}_t^{-1}\mathbf{S}_t^\top}}$. Taking $\norm{\tilde{\mathbf{x}}_t} \leq \epsilon'$, there exists
			\begin{dgroup*}
				\begin{dmath*}
					\gamma =  \inf_{\mathbf{x}}
					\dfrac{\norm{\mathbf{x}_t^\top \left(\mathbf{P}_t^{-1}\left(\mathbf{Q}_t-\mathbf{S}_t\mathbf{R}_t^{-1}\mathbf{S}_t^\top\right)\mathbf{P}_t^{-1} + \mathbf{H}_t^\top\mathbf{R}_t^{-1}\mathbf{H}_t - \kappa \epsilon'\right) \tilde{\mathbf{x}}_t}}
					{\norm{\tilde{\mathbf{x}}_t^\top\mathbf{P}_t^{-1}\tilde{\mathbf{x}}_t}}
					.
				\end{dmath*}
			\end{dgroup*}
			Moreover, since $\mathbf{x}_t$ is bounded, then there exists $\delta$ which is given by,
			\begin{dmath*}
				\delta = \sup_{\mathbf{x}_t} \trace{\left(\mathbf{G}(\mathbf{x}_t)^\ast - \mathbf{K}_t\mathbf{L}(\mathbf{x}_t)^\ast\right)^\top\mathbf{P}_t^{-1}\left(\mathbf{G}(\mathbf{x}_t) - \mathbf{K}_t \mathbf{L}(\mathbf{x}_t)\right)}.
			\end{dmath*}
			The result then follows immediately.
		\end{proof}

		In contrast to the proof of the stochastic convergence of a classical EKF given in \cite[Theorem 3.2]{Reif2000}, having $\norm{\tilde{\mathbf{x}}_0} \leq \epsilon'$ does not guarantee that in the future the estimation error will always remain in the region $\norm{\tilde{\mathbf{x}}_t} \leq \epsilon'$. Consequently, since the last dissipative inequality is only valid in the region $\norm{\tilde{\mathbf{x}}_t} \leq \epsilon'$ then we can only state the quadratic bound of the estimation error before it leaves the region $\norm{\tilde{\mathbf{x}}_t} \leq \epsilon'$.
		From the definition of the semi-norm in \eqref{eq:SemiNorm}, by \eqref{eq:normError} we obtain,
		\begin{dmath}
			\norm{\mathbf{x}_t-\hat{\mathbf{x}}_t } = \Paverage{\condExpect{\left(\mathbf{x}_t - \hat{\mathbf{x}}_t\right)^\dagger\left(\mathbf{x}_t - \hat{\mathbf{x}}_t\right)}}^{1/2}.
		\end{dmath}
		Consequently, the set all events such that the estimation error is greater than $\epsilon'$ is $\mathscr{Y}_t$ measurable, i.e. $\left\{\norm{\mathbf{x}_t-\hat{\mathbf{x}}_t } > \epsilon' \right\} \in \mathscr{Y}_t$. Let $\mathscr{E}_{\epsilon'} = \left\{\hat{\mathbf{x}}_t: \norm{\mathbf{x}_t - \hat{\mathbf{x}}_t} \leq \epsilon' \right\} \subset \mathscr{Y}_t$. Now, define the Markov time $\tau_{\epsilon'}\left(\hat{\mathbf{x}}_0\right) : \left\{\tau_{\epsilon'}\left(\hat{\mathbf{x}}_0\right) \leq t \right\} \in \mathscr{Y}_t$ to be the first time that the trajectory $\hat{\mathbf{x}}_t$ leaves $\mathscr{E}_{\epsilon'}$ given that it begins inside it. The following theorem describes the quantum EKF's estimation error quadratic bound.
		\begin{theorem}
			If the set of system observables $\mathbf{x}_t$ has an evolution given by the QSDE in \eqref{eq:QSDE} that satisfies Assumption \ref{asm:Assumptions} and if initially $\norm{\tilde{\mathbf{x}}_0} \leq \epsilon'$, where $\epsilon'$ is given in Lemma \ref{lem:Dissipativity}, then for all $t'  = \min \left(\tau_{\epsilon'}\left(\hat{\mathbf{x}}_0\right),t\right) \geq 0$
			\begin{dmath}
				\condExpectB{\tilde{\mathbf{x}}_{t'}^\top\tilde{\mathbf{x}}_{t'}}{\mathscr{Y}_{t'}} \leq \dfrac{\bar{p}}{\underline{p}}\condExpectB{\tilde{\mathbf{x}}_0^\top\tilde{\mathbf{x}}_0}{\mathscr{Y}_{t'}} e^{-\gamma {t'}} + \dfrac{\delta}{\gamma \underline{p}} \left(1 - e^{-\gamma t'}\right).
			\end{dmath}\label{thm:ConvergenceResult}
		\end{theorem}
		\begin{proof}
			The proof of this theorem makes use of the proof of the classical ultimate bound of quadratic moments given in \cite{zakai1967ultimate}, except that in this theorem, we replace the generator of the classical Markov process with the quantum Markov process generator (\emph{Lindblad} generator) in  \eqref{eq:QuantumMarkovProcessGenerator}. We also use the operator inequality with respect to the state $\mathbb{P}$ in $\mathscr{Y}_t$. First, we observe that $t'$ is a $\mathscr{Y}_{t'}$ measurable random variable \cite[Lemma 1.5]{liptser2001statistics}. By our definition of a Lyapunov positive operator in  \eqref{eq:Lyanpunov}, we may apply the It\^{o} formula, and therefore we obtain,
			\begin{dmath}
				\condExpectB{V\left(\tilde{\mathbf{x}}_{t'}\right)}{\mathscr{Y}_{t'}} = V\left(\tilde{\mathbf{x}}_0\right) + \int_{0}^{t'} \condExpectB{\mathcal{L}\left(V\left(\tilde{\mathbf{x}}_s\right)\right)}{\mathscr{Y}_{t'}}  ds. \label{eq:CondExpectOfV}
			\end{dmath}
			The equation above indicates that the expectation $\condExpectB{V\left(\tilde{\mathbf{x}}_{t'}\right)}{\mathscr{Y}_{t'}} $ is absolutely continuous in $t'$ since $\mathcal{L}\left(V\left(\tilde{\mathbf{x}}_s\right)\right)$ is Lebesgue integrable, and hence for almost all $s \geq 0$
			\begin{dmath*}
				\dfrac{d \condExpectB{V\left(\tilde{\mathbf{x}}_s\right)}{\mathscr{Y}_{t'}} }{ds} \leq -\gamma \condExpectB{V\left(\tilde{\mathbf{x}}_s\right)}{\mathscr{Y}_{t'}} + \delta.
			\end{dmath*}
			Multiplying by the factor $e^{\gamma s}$, we have
			\begin{dmath*}
				\dfrac{d \left(e^{\gamma s}\condExpectB{V\left(\tilde{\mathbf{x}}_s\right)}{\mathscr{Y}_{t'}} \right)}{ds} \leq \delta e^{\gamma s}.
			\end{dmath*}
			By  \eqref{eq:CondExpectOfV}, $e^{\gamma s}\condExpectB{V\left(\tilde{\mathbf{x}}_s\right)}{\mathscr{Y}_{t'}} $ is also absolutely continuous, and hence by integrating the last inequality, we obtain,
			\begin{dmath*}
				\condExpectB{V\left(\tilde{\mathbf{x}}_{t'}\right)}{\mathscr{Y}_{t'}} \leq \condExpectB{V\left(\tilde{\mathbf{x}}_0\right)}{\mathscr{Y}_{t'}} e^{-\gamma t'} + \dfrac{\delta}{\gamma} \left(1 - e^{-\gamma t'}\right).
			\end{dmath*}
			From \ref{asm:PAsm}, we obtain
			\begin{dmath*}
				\underline{p} \condExpectB{\tilde{\mathbf{x}}_{t'}^\top\tilde{\mathbf{x}}_{t'}}{\mathscr{Y}_{t'}} \leq \condExpectB{V\left(\tilde{\mathbf{x}}_{t'}\right)}{\mathscr{Y}_{t'}} \leq \condExpectB{V\left(\tilde{\mathbf{x}}_0\right)}{\mathscr{Y}_{t'}} e^{-\gamma t'} + \dfrac{\delta}{\gamma} \left(1 - e^{-\gamma t'}\right) \leq \bar{p}\condExpectB{\tilde{\mathbf{x}}_0^\top\tilde{\mathbf{x}}_0}{\mathscr{Y}_{t'}} e^{-\gamma t'} + \dfrac{\delta}{\gamma} \left(1 - e^{-\gamma t'}\right).
			\end{dmath*}
			Dividing the last result by $\underline{p}$ gives the desired result.
		\end{proof}

		\begin{remark}
			Using a treatment similar to \cite[Theorem 4.2]{Reif2000}, it can be shown that the condition in \ref{asm:PAsm} can be replaced with the uniform detectability of the pair $\mathbf{F}\left(\hat{\mathbf{x}}\right),\mathbf{H}\left(\hat{\mathbf{x}}\right), \forall \mathbf{x} \in \mathscr{Y}_t$ to obtain the same result in the Theorem \ref{thm:ConvergenceResult}, see \cite[Theorem 7]{Baras1988}.
		\end{remark}

		\subsection{Robust quantum nonlinear filter for a class of open quantum system with state-dependent noise}\label{sec:rqEKF}
		In the previous subsection, the quantum EKF is developed for a class of an open quantum systems where the noise variances are known. We also notice that the dissipation inequality given in Lemma \ref{lem:Dissipativity} is a very conservative condition and often rather difficult to validate. Moreover, for many open quantum systems, the variances of the system observables and the measurements are state-dependent. As an example, the covariance of photon counting is indeed a stochastic process, for which is a 'doubly stochastic Poisson' or Cox process \cite{Snyder1972,Segall1975}. In classical systems, doubly stochastic Poisson processes have been treated in a variety of ways, e.g. by solving a conditional probability density evolution, or a filtered martingale problem approach \cite{Snyder1972,segall1975nonlinear,ceci2012nonlinear,ceci2014zakai}.
		\\
		In this subsection, we will show how to modify the quantum EKF to include open quantum systems and measurements with state-dependent noise variance and cross correlation, while at the same obtain a stronger condition of convergence than Lemma \ref{lem:Dissipativity}. The treatment we present here uses Riccati differential equation shaping, which is to some extent similar to the treatment of the deterministic nonlinear observer in \cite{reif1998ekf,reif1999nonlinear}.
		\\
		We will still use the same filter dynamics as in  \eqref{eq:QuantumEKF}. We will also use the same notation $\mathbf{P}_t$ to denote a positive definite matrix of operators, whose dynamics is shaped to achieve  robustness of the estimation error dynamics.
		\\
		Since the variance and covariance are functions of $\mathbf{x}_t$, we will use the estimates of $\mathbf{R}_t$ and $\mathbf{S}_t$  in the Kalman filter, and therefore we have
		\begin{align}
		\mathbf{K}_t =& \left[\mathbf{P}_t{\mathbf{H}(\hat{\mathbf{x}}_t)}^\top + \mathbf{S}(\hat{\mathbf{x}}_t)\right]{\mathbf{R}(\hat{\mathbf{x}}_t)}^{-1}. \label{eq:RobustKalmanGain}
		\end{align}
		Furthermore, we shape the Riccati differential equation, where for given $\mu, \lambda > 0$, with $\hat{\mathbf{Q}} > \mu \mathbf{I}$,
		\begin{align}
		\dfrac{d\mathbf{P}_t}{dt} = & \mathbf{F}(\hat{\mathbf{x}}_t)\mathbf{P}_t+\mathbf{P}_t\mathbf{F}(\hat{\mathbf{x}}_t)^\top + \hat{\mathbf{Q}} + \lambda {\mathbf{P}_t}^2 - \mathbf{K}_t {\mathbf{R}(\hat{\mathbf{x}}_t)} \mathbf{K}_t^\top. \label{eq:RobustRiccati}
		\end{align}
		We state our result in the following theorem,
		\begin{theorem}\label{thm:RobustEKF}
			Consider an open quantum system QSDE given in \eqref{eq:QSDE}, the measurement given in \eqref{eq:Measurements}, and the quantum EKF given in \eqref{eq:QuantumEKF}. Also assume Assumptions \ref{asm:ResidualAsm}, \ref{asm:Hbounded}, and \ref{asm:PAsm}, and $\norm{\tilde{\mathbf{x}}_t} \leq \epsilon, \forall t \geq 0$, where  $\epsilon = \min \left(\epsilon_f,\epsilon_h\right)$ in \ref{asm:ResidualAsm}. Then there exist $\mu,\lambda>0$ in  \eqref{eq:RobustRiccati}  and $\gamma, \delta >0$ such that the Lyapunov positive operator given in \eqref{eq:Lyanpunov}, satisfies the dissipation inequality below
			\begin{align}
			\mathcal{L}\left(V\right)  \leq & -\gamma V(\tilde{\mathbf{x}}_t)  + \delta. \label{eq:RobustDissipativityDifferential}
			\end{align}
		\end{theorem}
		\begin{proof}
			Using the same argument as in Lemma \ref{lem:Dissipativity},
			\begin{dmath*}
				\mathcal{L}\left(V\right)  =  \tilde{\mathbf{x}}_t^\top\mathbf{P}_t^{-1}\left[\left(\mathbf{F}(\hat{\mathbf{x}}_t) - \mathbf{K}_t\mathbf{H}(\hat{\mathbf{x}}_t)\right)\tilde{\mathbf{x}}_t  + \mathbf{r} \right]+\left[\left(\mathbf{F}(\hat{\mathbf{x}}_t) - \mathbf{K}_t\mathbf{H}(\hat{\mathbf{x}}_t)\right)\tilde{\mathbf{x}}_t  + \mathbf{r} \right]^\top\mathbf{P}_t^{-1}\tilde{\mathbf{x}}_t  +  \trace{\left(\mathbf{G}(\mathbf{x}_t)^\ast - \mathbf{K}_t\mathbf{L}(\mathbf{x}_t)^\ast\right)^\top\mathbf{P}_t^{-1}\left(\mathbf{G}(\mathbf{x}_t) - \mathbf{K}_t \mathbf{L}(\mathbf{x}_t)\right)} - \tilde{\mathbf{x}}_t^\top\mathbf{P}_t^{-1}\dfrac{\partial \mathbf{P}_t}{\partial t}\mathbf{P}_t^{-1}\tilde{\mathbf{x}}_t.
			\end{dmath*}
			With the same $\delta$ as in the Lemma \ref{lem:Dissipativity}, then from Lemma \ref{lem:NonlinearityBound}, we can select sufficiently large $\epsilon = \min\left(\epsilon_f,\epsilon_h \right)$, and use \eqref{eq:errorDynamic}, \eqref{eq:RobustKalmanGain}, and \eqref{eq:RobustRiccati} to obtain
			\begin{dmath*}
				\mathcal{L}\left(V\right) \leq   -\tilde{\mathbf{x}}_t^\top\mathbf{P}_t^{-1}\mathbf{U}\mathbf{P}_t^{-1}\tilde{\mathbf{x}}_t + \delta_t ,
			\end{dmath*}
			with
			\begin{dmath*}
				\mathbf{U} =  \mathbf{P}_t\left[\mathbf{H}\left(\hat{\mathbf{x}}\right){\mathbf{R}(\hat{\mathbf{x}}_t)}^{-1}{\mathbf{H}\left(\hat{\mathbf{x}}\right)}^\top + \left(\lambda - \kappa \epsilon\right)\mathbf{I}\right]\mathbf{P}_t + \left[\hat{\mathbf{Q}} -\mathbf{S}(\hat{\mathbf{x}}_t) {\mathbf{R}(\hat{\mathbf{x}}_t)}^{-1}{\mathbf{S}(\hat{\mathbf{x}}_t)}^\top\right],
			\end{dmath*}
			where $\mathbf{I}$ is the identity matrix. Furthermore, selecting
			\begin{dmath*}
				\hat{\mathbf{Q}}  = \mu\mathbf{I} + \mathbf{S}(\hat{\mathbf{x}}_t) {\mathbf{R}(\hat{\mathbf{x}}_t)}^{-1}{\mathbf{S}(\hat{\mathbf{x}}_t)}^\top,
			\end{dmath*}
			and $\lambda > \kappa \epsilon$, from the fact that ${\mathbf{R}(\hat{\mathbf{x}}_t)}^{-1} > 0$, we infer that $\mathbf{U} > 0$. Furthermore, with $\underline{\mathbf{U}} = \inf_{\mathbf{x}}\mathbf{U}$, $\gamma$ is given by
			$\gamma = \frac{\norm{\underline{\mathbf{U}}}}{\bar{p}} $, which completes the proof.

		\end{proof}
		The quadratic bound on the estimation error follows immediately by similar assertions in the Theorem \ref{thm:ConvergenceResult}, by replacing $\epsilon'$ with $\epsilon$.
		\\
		We notice that we could independently select $\epsilon$ sufficiently large due to extra freedom introduced by $\mu$ and $\lambda$. These two parameters remove the positive definiteness restriction of $\mathbf{Q}_t - \mathbf{S}_t\mathbf{R}_t^{-1}\mathbf{S}_t^\top$, as in Assumption \ref{asm:QRSAsm}.
		$\lambda$ is selected to directly dominate the nonlinearity effect $\kappa$ in the filter that can lead to the divergence of the estimation errors. On the other hand, $\mu$ increases the convergence, ensuring the inequality in $\mathbf{P}_t$ \ref{asm:PAsm}, while at the same time increases the noise level of the estimation.

		\section{Application Examples}
		In this section, we consider some examples of the application of the quantum EKF. We begin with the estimation of position and momentum quadratures of two optical cavity modes with Kerr nonlinearities and subject to homodyne detection. In this example the variance matrices $\mathbf{Q},\mathbf{R}$ and $\mathbf{S}$ are  constant matrices. Next, we show the estimation of the position and momentum quadratures in an optical cavity subjected to simultaneous homodyne and photon counting detections, which corresponds to the case where the variance matrices $\mathbf{R}$ and $\mathbf{S}$ are a state-dependent matrices. The constant $\hbar$ is assumed to be one. 	As before, we will denote the isomorphically transformed operator in $\mathscr{Y}_t$, $\tau\left(\cdot\right)_{t,\omega} \in \mathbb{C}$ as $(\cdot)_{t,\omega}$.

		\subsection{Estimating the quadratures of multiple optical modes with a Kerr Hamiltonian}
		\begin{figure*}

			\subfloat[$t_{SME}/t_{qEKF}$]
			{\label{fig:t_ratio}\includegraphics[width=0.55\textwidth]{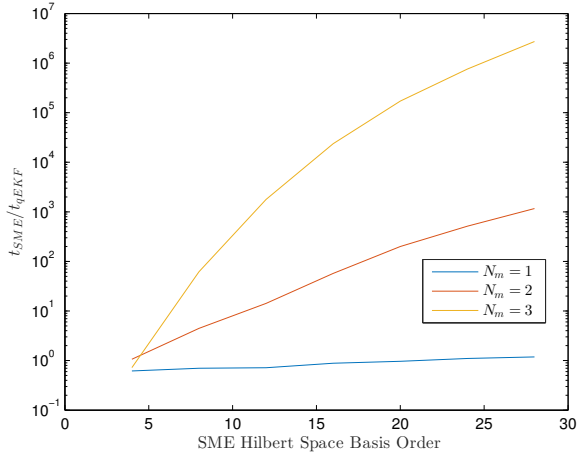}}
			\subfloat[MISE comparison ]
			{\label{fig:MISE_TwoModes}\includegraphics[width=0.55\textwidth]{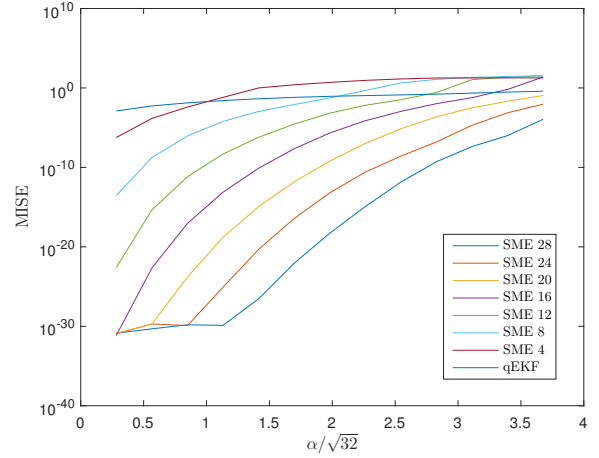}}\\
			\subfloat[$ q_1 $ Error ]
			{\label{fig:TwoModesInitialError}\includegraphics[width=0.55\textwidth]{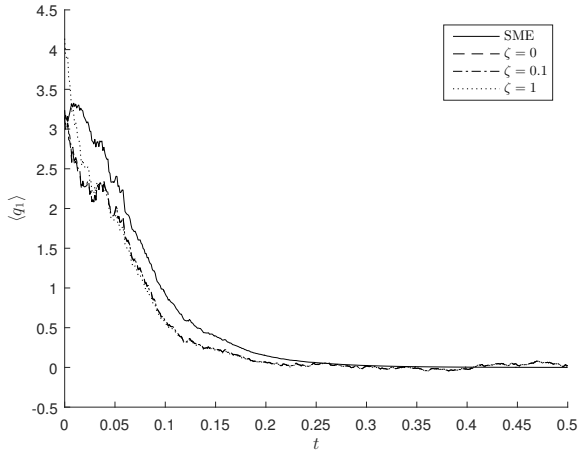}}
			\caption{(color online) Application of the quantum EKF to estimate quadratures of two cavity modes with Kerr Hamiltonians. Here, we use a cavity with damping constant $ \gamma_i = 32$ and Kerr nonlinearity constant $ \chi_i = 0.3 \pi$. The Hilbert space basis for each modes is of order to 32. \Fref{fig:t_ratio} shows the ratio of one time-step computational time between the SME and the quantum EKF for various number of modes, and order of the Hilbert space basis used for solving the SME. The ratio increases dramatically when the number of modes increases.	\Fref{fig:MISE_TwoModes} shows a comparison of the mean integral of squared quadratic error, $\frac{1}{T}(\int_0^T \tilde{\mathbf{x}}_{t,\omega}^\top \tilde{\mathbf{x}}_{t,\omega} dt)^{1/2} $, for the quantum EKF and SME estimation, with various different orders of the Hilbert space basis, against different initial coherent amplitudes $\alpha$. The error $\tilde{\mathbf{x}}_{t,\omega}$ is approximated by subtracting the estimate of the quantum EKF and the SME against the estimate of the SME with the largest Hilbert space basis $(32)$. \Fref{fig:TwoModesInitialError}  shows the estimation of the first mode quadrature with different initial errors. The initial value of the quantum EKF is given by $\hat{\mathbf{x}}_{0,\omega} = \hat{\mathbf{x}}_{SME (0,\omega)} + \zeta [1\;-1\;-1\;1 ]^\top$.  $\zeta$ corresponds to the magnitude of the initial estimation error. In this simulation, the trajectories of the SME and the quantum EKF estimation are generated from identical measurement records.}
			\label{fig:TwoModesMix}
		\end{figure*}

		In this example, we would like to estimate the quadratures of two cavity modes $\mathbf{x}_{i,t} = \left[q_i \; p_i\right]^\top, i=1,2$. Each mode has an identical Kerr Hamiltonian and a direct coupling Hamiltonian. For $n$ cavity modes, the direct coupling Hamiltonian is given by \cite{Wiseman1994}, $\mathbb{H}_{int} = \sum\limits_{\substack{i,j=1\\i\neq j}}^{n} i  \sqrt{\gamma_{ij}}\left(a_j a^\dagger_i - a_i a^\dagger_j\right) = \sum\limits_{\substack{i=1\\j=i+1}}^{n} \mathbf{x}_{i,t}^\top \breve{\mathbf{S}}_{ij} \mathbf{x}_{j,t}$. The parameters for the optical cavity, are given by $\mathbb{S} = \mathbf{I} , \mathbb{L}_i =  \sqrt{\gamma}a_i = \mathbf{C}_i^\top \mathbf{x}_{i,t}$, and
		\begin{dmath*}
			\mathbb{H} = \sum\limits_{i=1}^{n} \chi_i {a_i^\dagger}^2 a_i^2  +\sum\limits_{\substack{i=1\\j=i+1}}^{n} \mathbf{x}_{i,t}^\top \breve{\mathbf{S}}_{ij} \mathbf{x}_{j,t}.
		\end{dmath*}
		Let $\bm{\Sigma} = \begin{bmatrix}
		0 && 1 \\ -1 && 0
		\end{bmatrix}$. The quantum Markovian generator for $\mathbf{x}_t$, is given by $ \mathbf{f}_i(\mathbf{x}_t) = \sum\limits_{j=0}^{n} \mathbf{A}_{ij}\mathbf{x}_{j,t} + \mathbf{f}_{i,Kerr}(\mathbf{x}_{i,t})$, while $\mathbf{G}_i =  i \bm{\Sigma} \mathbf{C}_i$. Now $\mathbf{A}_{ij}$ is given by the following matrices
		\begin{dmath*}
			\mathbf{A}_{ij} =
			\begin{cases}
				- \bm{\Sigma} \Im\left(\mathbf{C}_i\mathbf{C}_i^\dagger\right) &, i=j \\
				\bm{\Sigma} \breve{\mathbf{S}}_{ij} &, i \neq j
			\end{cases}.
		\end{dmath*}
		$\mathbf{f}_{i,Kerr}(\mathbf{x}_t)$ is the nonlinearity from the Kerr Hamiltonian, we have
		\begin{dgroup*}
			\begin{dmath*}
				\mathbf{f}_{i,Kerr} =   \frac{\chi}{4} \begin{bmatrix}
					4p_{i,t}^3+p_{i,t}q_{i,t}^2+q_{i,t}^2p_{i,t}+2q_{i,t}p_{i,t}q_{i,t} - 8p_{i,t}\\
					-\left(4q_{i,t}^3+q_{i,t}p_{i,t}^2+p_{i,t}^2q_{i,t}+2p_{i,t}q_{i,t}p_{i,t} - 8q_{i,t}\right)
				\end{bmatrix}
			\end{dmath*},
			\begin{dmath*}
				\mathbf{F}_i\left(\mathbf{x}\right) = \chi \begin{bmatrix}
					q_{i,t}p_{i,t}+p_{i,t}q_{i,t} && 3p_{i,t}^2+0.5\left(q_{i,t}^2+p_{i,t}^2\right)-2\\
					-(3q_{i,t}^2+0.5\left(q_{i,t}^2+p_{i,t}^2\right) - 2) && -(q_{i,t}p_{i,t}+p_{i,t}q_{i,t})
				\end{bmatrix}.
			\end{dmath*}
		\end{dgroup*}
		We consider homodyne detection on each cavity mode as measurement such that $\mathbf{E}=\mathbf{I}, \mathbf{N}=0$. From \Fref{fig:TwoModesInitialError}, it can be seen that under a small error at time zero, the quantum EKF still maintains the estimation error to be bounded, which shows the convergence of the quantum EKF.
		\\
		We can observe from \Fref{fig:TwoModes} that the area of mean and one standard deviation of the quantum EKF on the $q_1$,$p_1$,$q_2$ and $p_2$ are quite narrow and contained in the estimation of SME. In contrast, the quantum KF estimation has a substantially larger deviation from the truncated SME mean although it has a narrower standard deviation range.
		\\
		Regarding the computational time, with a Hilbert space basis of order 20 on each mode, the SME requires nearly 200 times that of the quantum EKF. The SME will require 170,000 times that of the quantum EKF if the number of modes equals to three, see also \Fref{fig:t_ratio}. From \Sref{sec:Complexity}, generally for every single mode increase in the quantum systems, the required time of the SME single step computation will be increased to by a factor of $ N_s^3$.\\
		In terms of estimation error, the quantum EKF can mantain a fairly low slope of the mean integral of squared quadratic error, $\frac{1}{T}(\int_0^T \tilde{\mathbf{x}}_{t,\omega}^\top \tilde{\mathbf{x}}_{t,\omega} dt)^{1/2} $ regardless of the initial state of the estimated quantum system. This is another benefit, since for the SME, the MISE could make a huge difference if the initial state needs a higher order Hilbert space basis, see \Fref{fig:MISE_TwoModes}.

		\begin{figure*}
			\subfloat[$ q_1 $ averaged]{\label{fig:ekfQ_MultiModes1}\includegraphics[trim =   5mm 10mm 5mm 10mm, clip,width=0.33\textwidth]{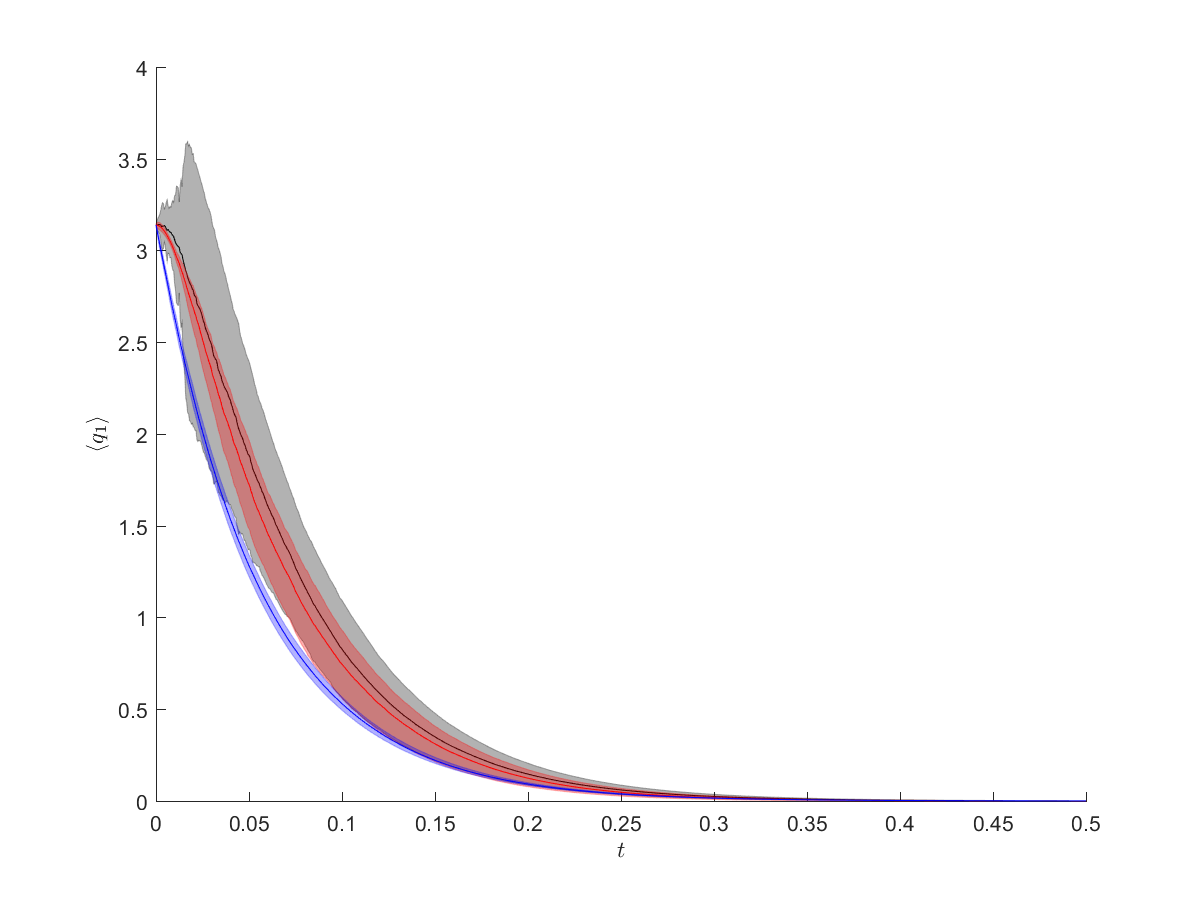}}
			\subfloat[$ p_1 $ averaged]{\label{fig:ekfP_MultiModes1}\includegraphics[trim =   5mm 10mm 5mm 10mm, clip,width=0.33\textwidth]{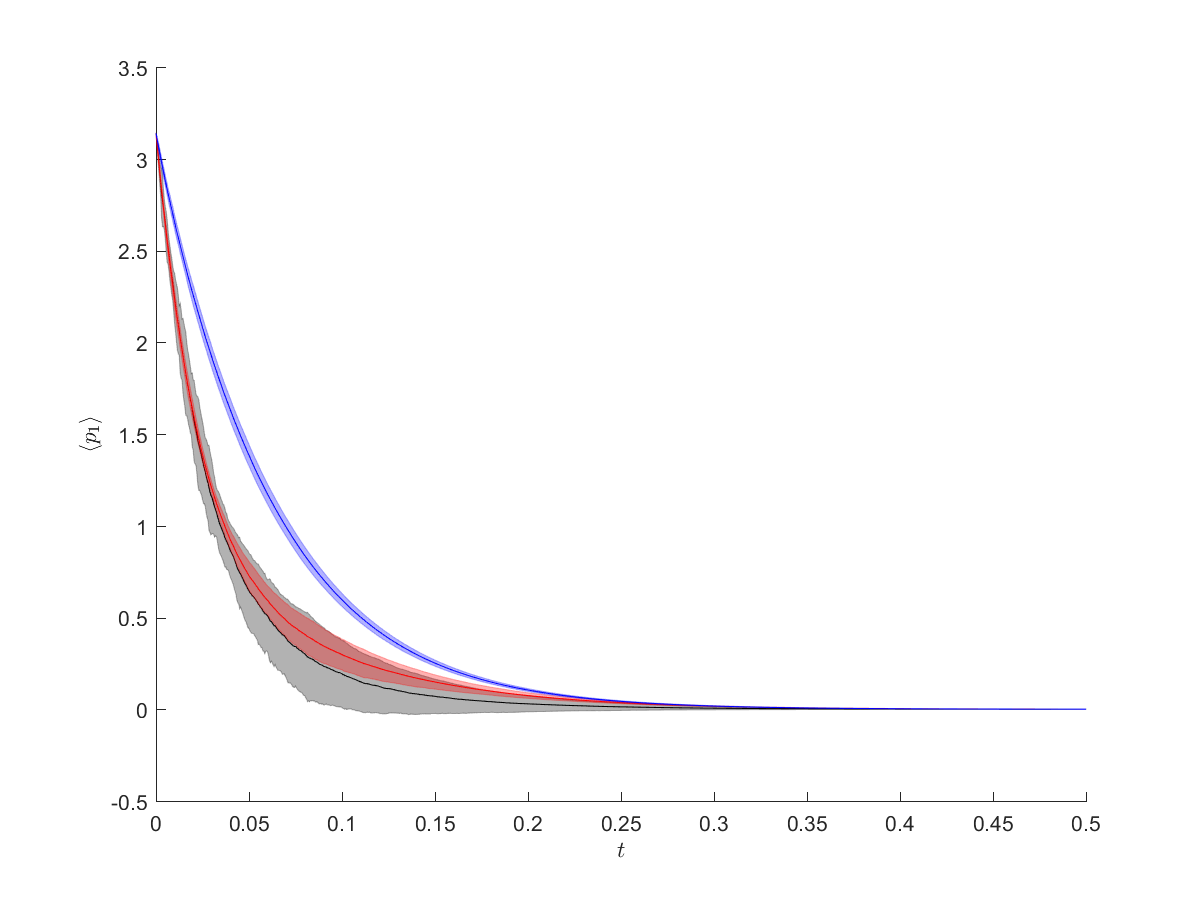}}
			\subfloat[Phase space]{\label{fig:ekfPhase_Averaged1}\includegraphics[trim =   5mm 10mm 5mm 10mm, clip,width=0.33\textwidth]{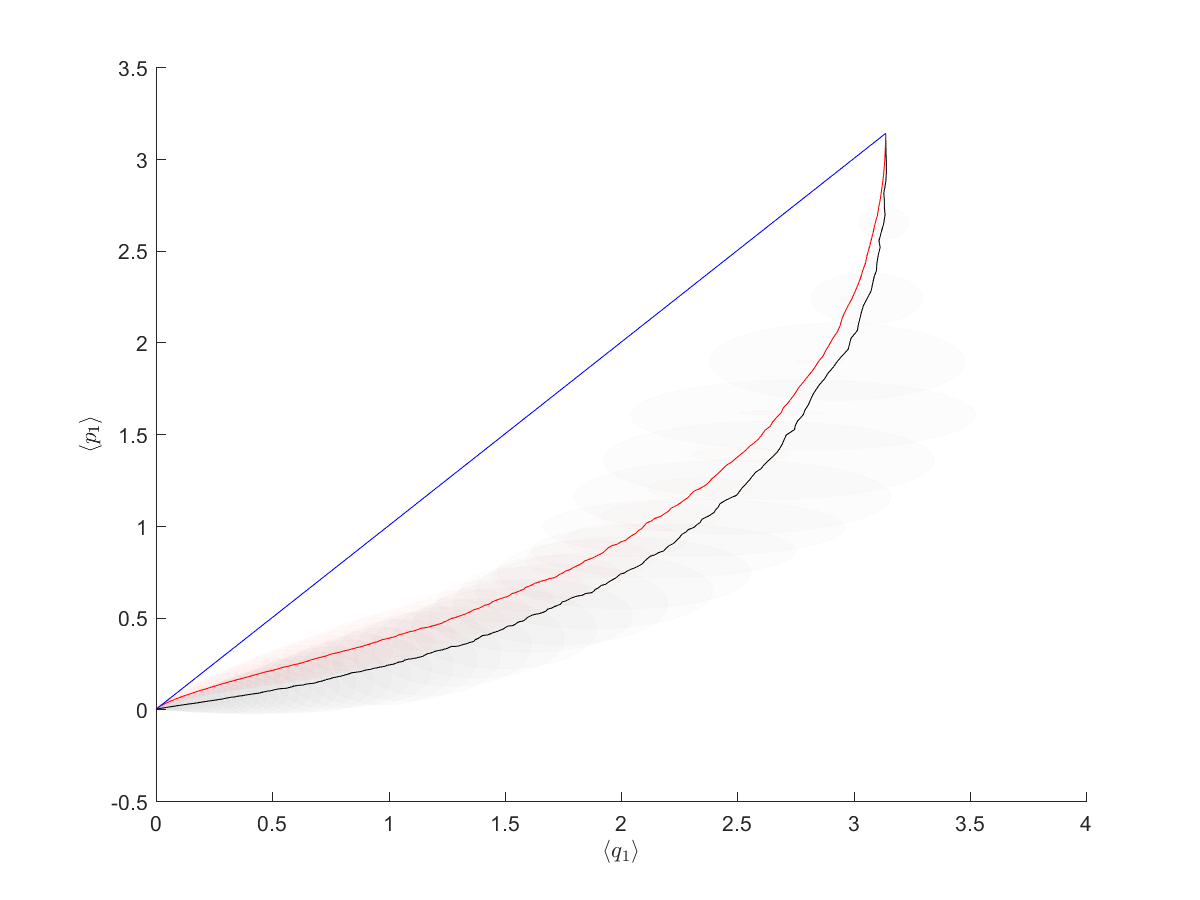}}\\
			\subfloat[$ q_2 $ averaged]{\label{fig:ekfQ_MultiModes2}\includegraphics[trim =   5mm 10mm 5mm 10mm, clip,width=0.33\textwidth]{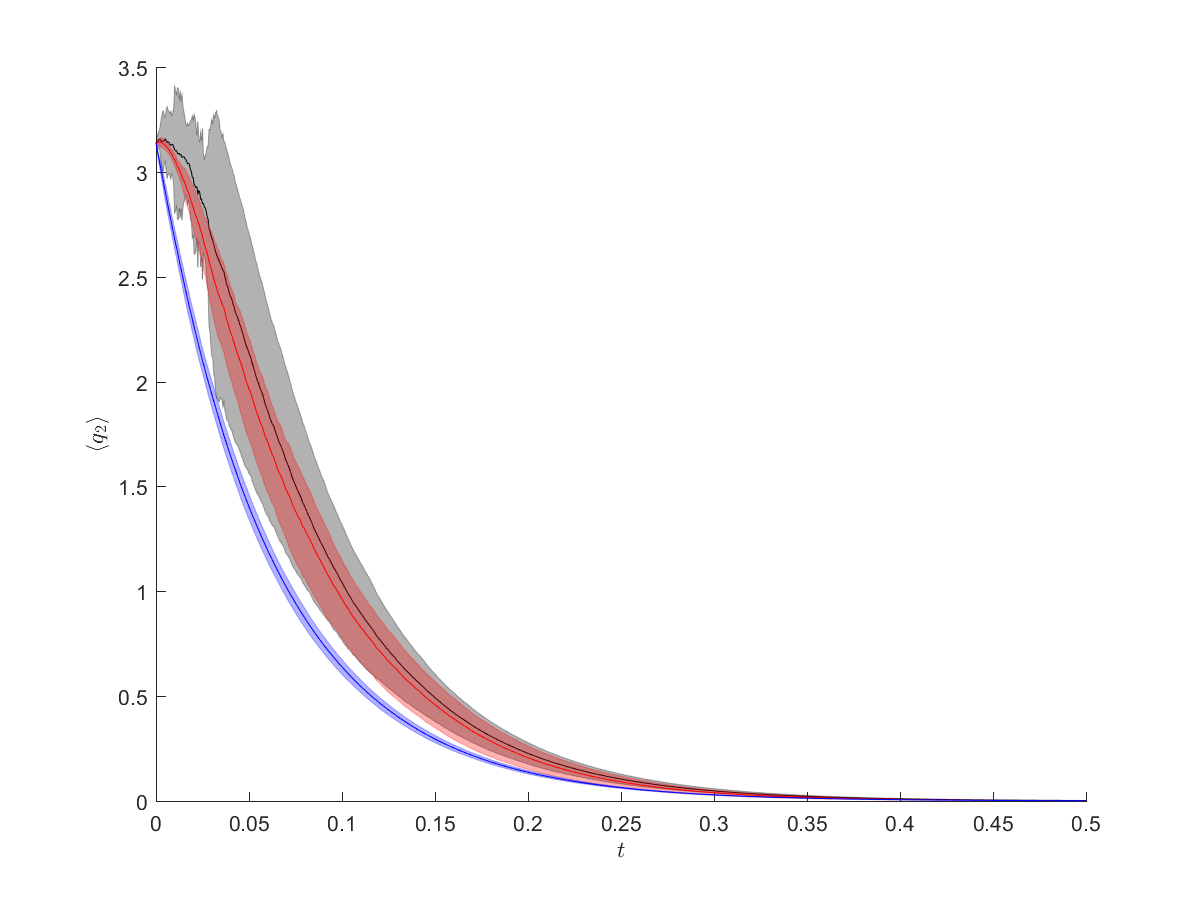}}
			\subfloat[$ p_2 $ averaged]{\label{fig:ekfP_MultiModes2}\includegraphics[trim =   5mm 10mm 5mm 10mm, clip,width=0.33\textwidth]{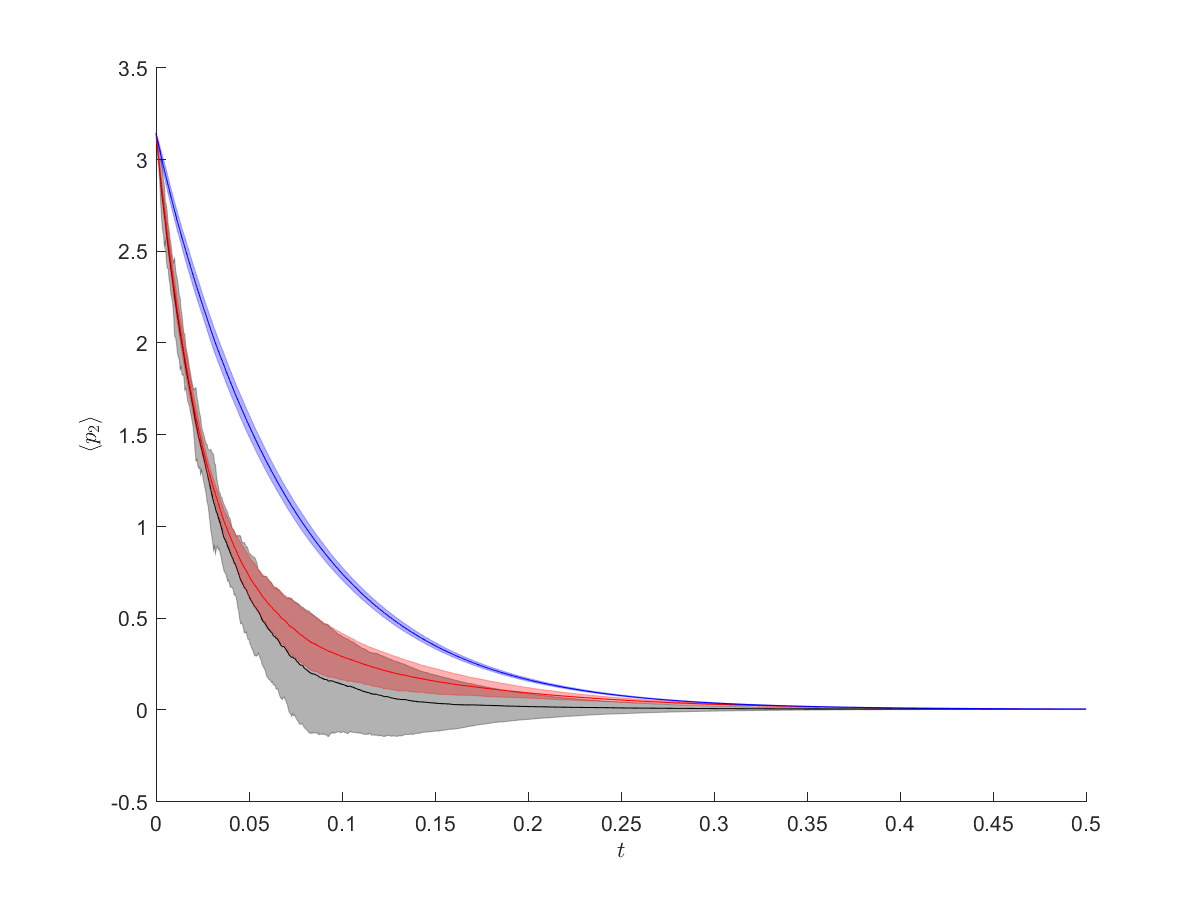}}
			\subfloat[Phase space]{\label{fig:ekfPhase_Averaged2}\includegraphics[trim =   5mm 10mm 5mm 10mm, clip,width=0.33\textwidth]{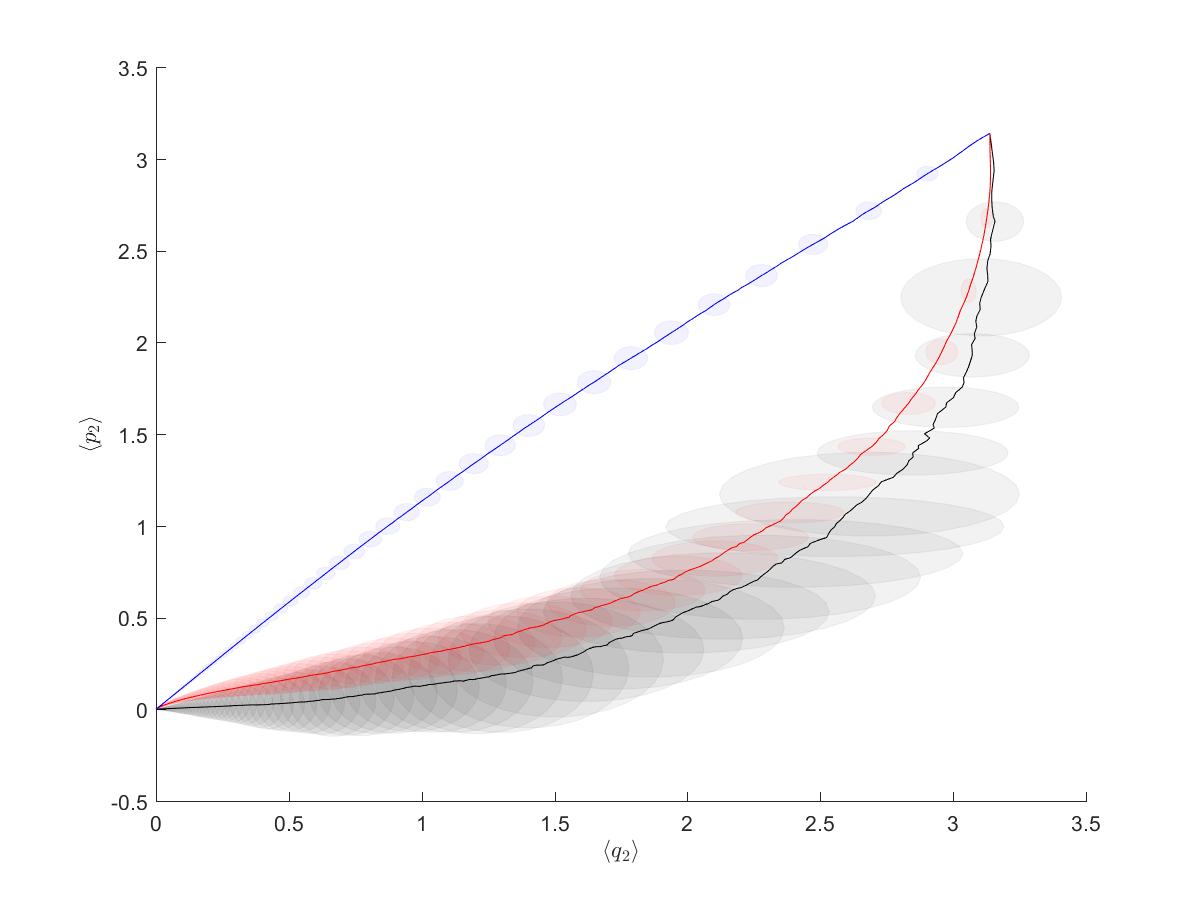}}

			\caption{(color online) As in \Fref{fig:TwoModesMix}, this simulation shows a comparison of the quantum EKF, the SME, and the quantum Kalman Filter (KF) estimation results.  The black line with a shaded area  corresponds to the mean with one standard deviation range from 100 Monte Carlo trials of the SME's estimation (Hilbert space basis of order 20 for each mode). The red line with a shaded area is the quantum EKF and the blue line is the estimation of the quantum KF where we simply ignore the nonlinearity. Initial errors are set to zero for these simulations.}
			\label{fig:TwoModes}
		\end{figure*}

		\subsection{Estimating the quadratures of an optical cavity subject to simultaneous homodyne detection and photon counting}
		\begin{figure}
			\centering
			\includegraphics[width=0.5\textwidth]{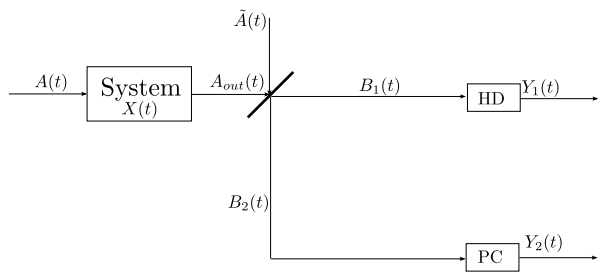}
			\caption{Simultaneous photon counting and homodyne detection at either output port of a beam splitter in a quantum optics experiment, and the corresponding quantum network depiction of the quantum optics setup, \cite{MuhammadF.Emzir2015}.}
			\label{fig:DoubleMeasurement}
		\end{figure}

		\Fref{fig:DoubleMeasurement} shows an  optical cavity with the simultaneous homodyne detection and photon counting setup as considered in \cite{MuhammadF.Emzir2015}. The involvement of photon counting in this optical setup makes the covariance matrix of the measurement  $\mathbf{R}$  state-dependent.
		\\
		Let $\mathcal{G}_1$ be our system of interest, with parameters $(\mathbb{1},\sqrt{\gamma} a,\mathbb{H})$. The vacuum noise is concatenated into our system by $\mathcal{G}_2$, whose parameters are $(\mathbb{1} ,0,0)$. The beam splitter is given by $\mathcal{G}_3$, with parameters $(\mathbb{S},0,0)$. By taking the series and the concatenation products \cite{gough2009series}, the parameters of the composite quantum system in \Fref{fig:DoubleMeasurement} are given by $\mathcal{G} = \left(\mathcal{G}_1 \boxplus \mathcal{G}_2\right)\rhd\mathcal{G}_3$,
		with $\left(\mathbb{S},\mathbb{S}\begin{pmatrix}
		\sqrt{\gamma} a\\0
		\end{pmatrix},\mathbb{H}\right)$.
		\\
		The beam splitter matrix $\mathbb{S}$, bath coupling $\mathbb{L}$ and Hamiltonian $\mathbb{H}$ are given by the following
		\begin{dgroup*}
			\begin{dmath*}
				{\mathbb{S} = \begin{bmatrix}\sqrt{1 - r^2} & i r\\ i r & \sqrt{1 - r^2} \end{bmatrix},\; r\geq 0, \;
					\mathbb{L} = \mathbb{S}\mathbf{C}^\top\mathbf{x}_t },
			\end{dmath*}
			\begin{dmath*}
				\mathbb{H} = i\left(\eta^\ast a^2 - \eta{a^\dagger}^2\right), { \eta \in \mathbb{C}}.
			\end{dmath*}
		\end{dgroup*}
		Figures \ref{fig:x_1trial}-\ref{fig:x_2trial} show a sample of Monte Carlo simulation of the SME (black), and the extension of the quantum EKF developed in \Sref{sec:rqEKF} (rqEKF) (blue). In this trial, the system's density operator is a superposition between a coherent and a Fock state. It is clear from Figs. \ref{fig:x_1trial}-\ref{fig:x_2trial} that rqEKF estimation gives a good approximation to the solution of the SME after a transient period. In the event of photon detection however, the rqEKF tends to have a slightly higher instantaneous jumps compared to the SME.
		\\
		Figures \ref{fig:x_1}-\ref{fig:x_2} shows how the rqEKF performs against the truncated SME in terms of the average estimate of 100 Monte Carlo trials and their one standard average region. We first observe that the area of mean and one stdev of the rqEKF for $q_t$ and $p_t$ gives a qualitatively good approximation after a small transient time similar to those of the SME.
		\begin{figure*}
			\centering
			\subfloat[$q$]{\label{fig:x_1}\includegraphics[trim =   5mm 10mm 5mm 10mm, clip, width=0.5\textwidth]{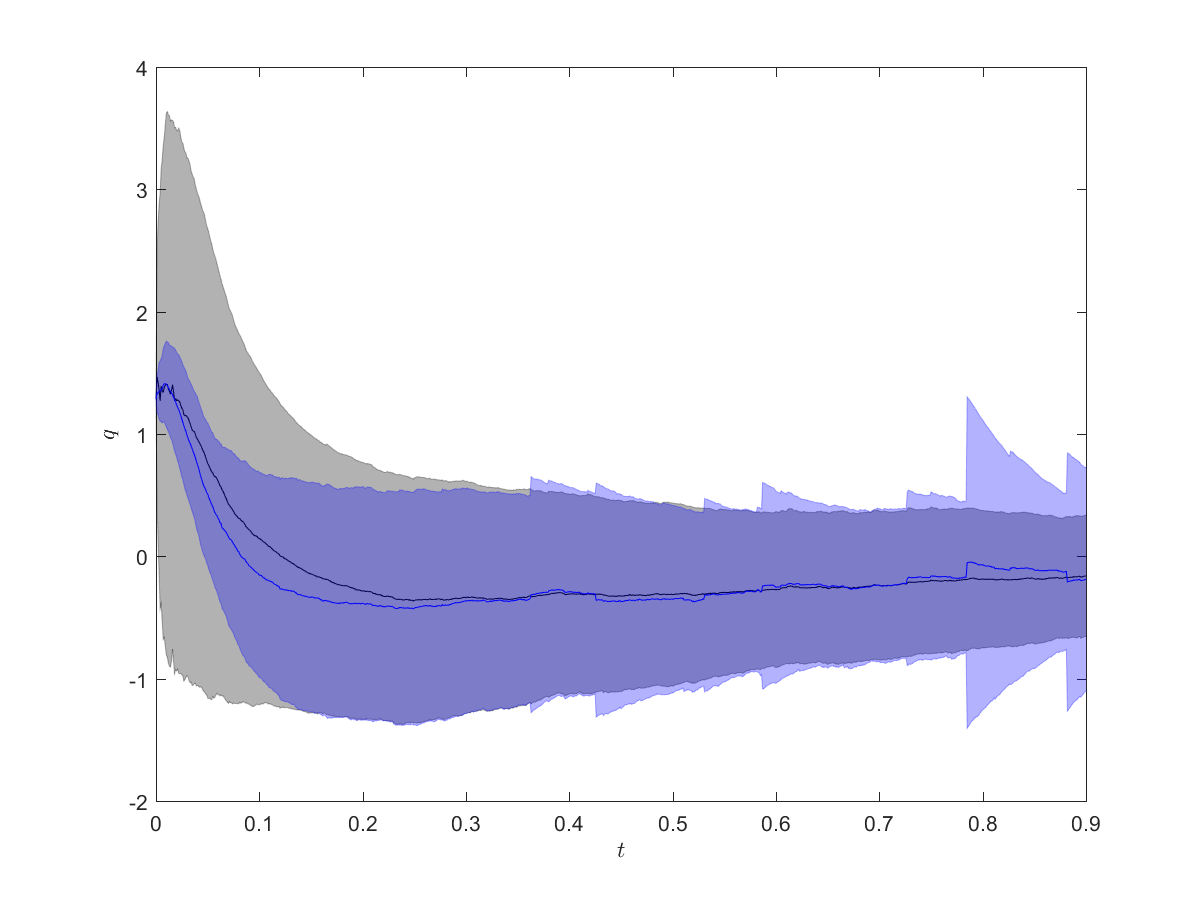}}
			\subfloat[$p$]{\label{fig:x_2}\includegraphics[trim =   5mm 10mm 5mm 10mm, clip, width=0.5\textwidth]{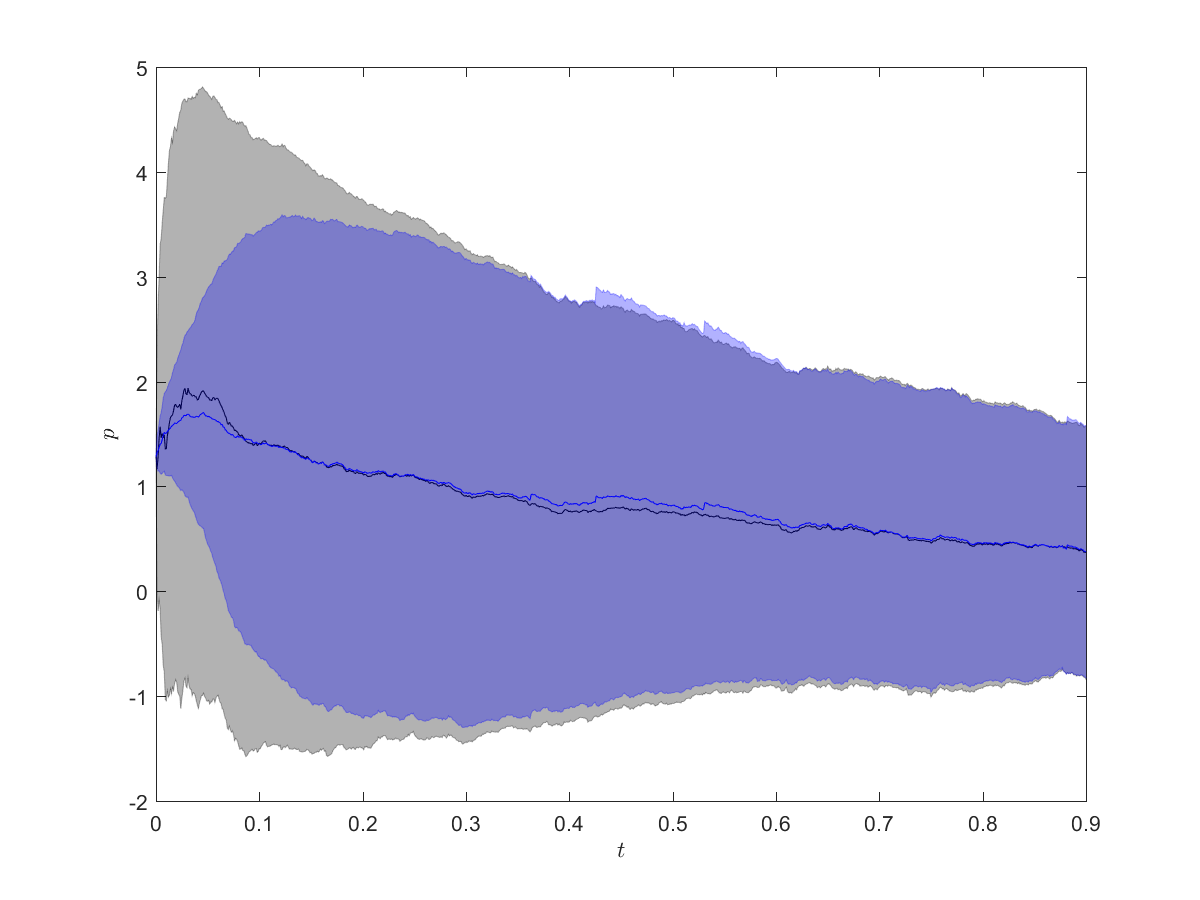}}\\
			\subfloat[$q$-trial]{\label{fig:x_1trial}\includegraphics[width=0.5\textwidth]{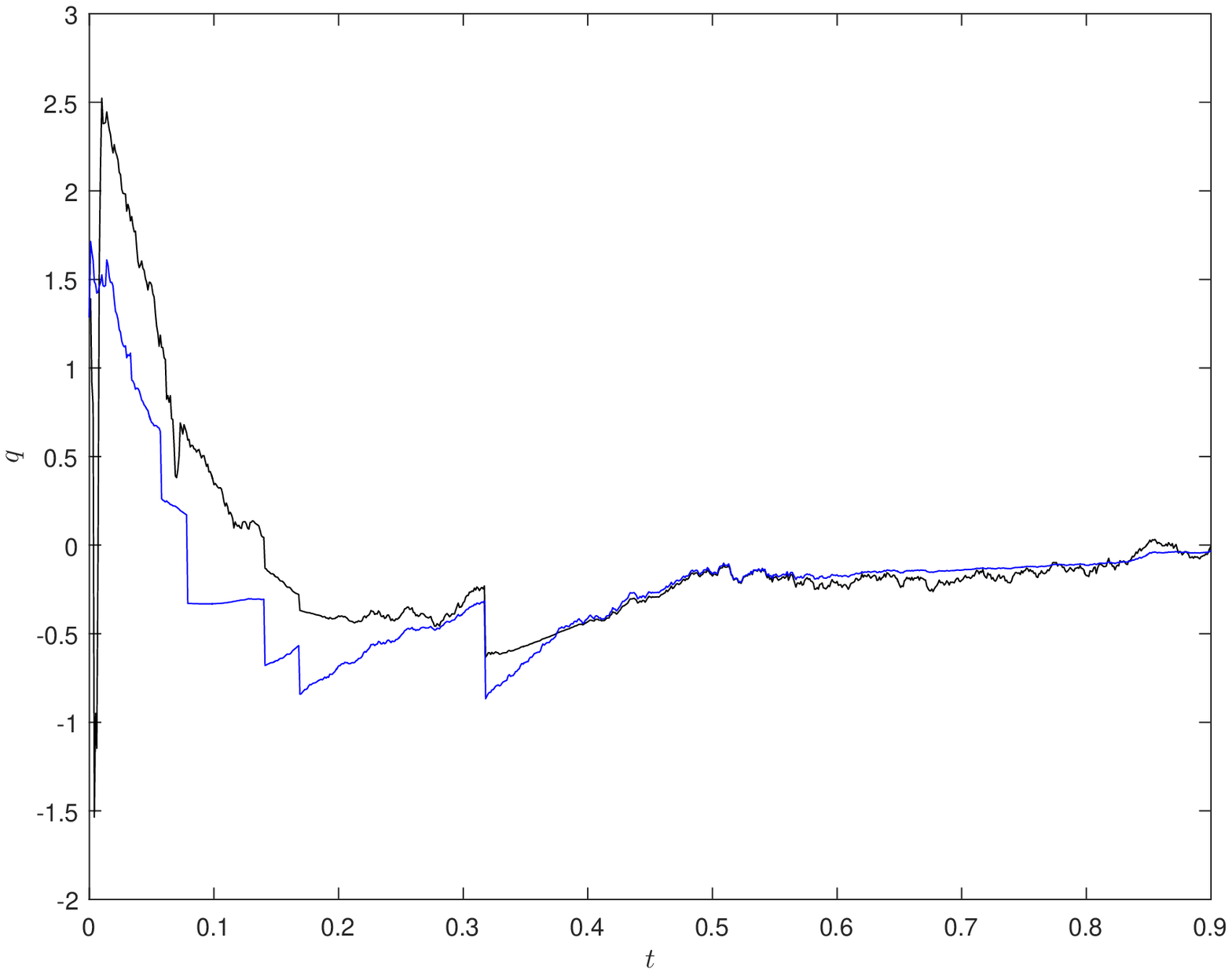}}
			\subfloat[$p$-trial]{\label{fig:x_2trial}\includegraphics[width=0.5\textwidth]{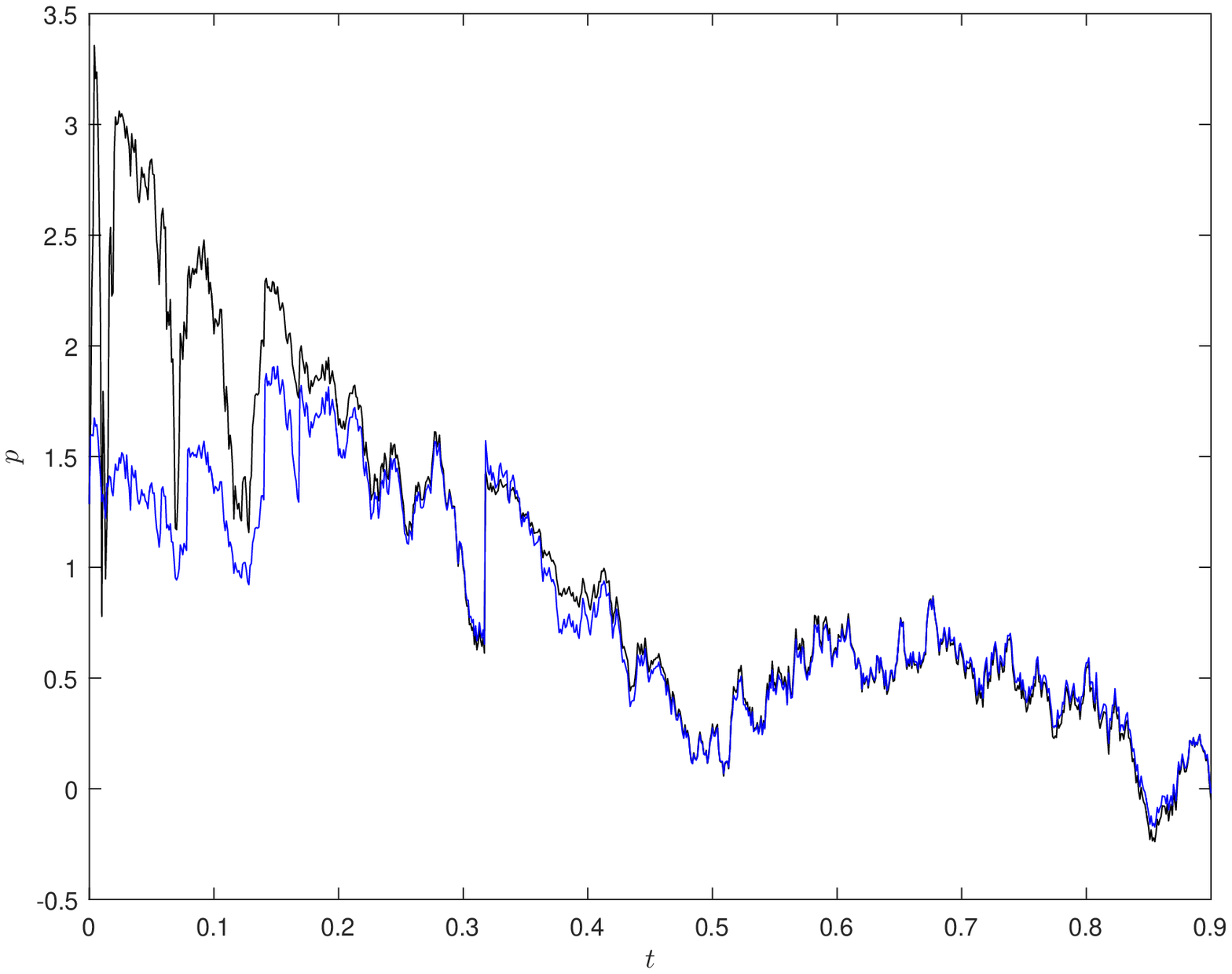}}
			\caption{(color online) These figures show the estimation of the optical cavity quadratures subject to simultaneous homodyne detection and photon counting. The system's Hilbert space basis (N) is of order $40$. The system's initial vector state (unnormalized) is $\ket{\tilde{\psi}}_0 = 0.5 \ket{n}+ 0.5 \ket{\alpha}$, with $n = N/2, \alpha = \sqrt{N/2}\exp(i\pi/4)$, while the reflectivity factor $r^2$  is 0.5. Figures \ref{fig:x_1}-\ref{fig:x_2} show the mean and one standard deviation range of 100 Monte Carlo trajectories of the SME (black) and the robust quantum EKF. Figures \ref{fig:x_1trial}-\ref{fig:x_2trial}, show $q$ and $p$ values from a sample of Monte Carlo trials.}
			\label{fig:SimulationOPO}
		\end{figure*}

		\section{Conclusion}
		In this article we have developed a quantum EKF for a class of nonlinear QSDEs describing an open quantum system subject to measurement. We derived a sufficient condition for the quantum EKF to achieve local quadratic exponential convergence in the estimation error. We also extended the quantum EKF to the class of  quantum systems and measurements with state-dependent covariance matrices. Finally, we have illustrated via two examples the effectiveness of the quantum EKF approximation.

		\section{References}

		\bibliographystyle{iopart-num}
		\bibliography{ReferenceAbbrv}
\end{document}